\documentclass[a4paper,12pt,twoside,draft]{article}
\usepackage{amsmath,amssymb,amsthm}

\evensidemargin = \oddsidemargin

\newcommand{\R}{\mathbb{R}}

\newcommand{\C}{\mathbb{C}}

\newcommand{\al}{\alpha}
\newcommand{\ov}{\overline}

\newcommand{\ve}{\varepsilon}

\newcommand{\om}{\omega}
\newcommand{\Om}{\Omega}
\newcommand{\lam}{\lambda}
\newcommand{\na}{\nabla}
\newcommand{\si}{\sigma}
\newcommand{\Si}{\Sigma}

\newcommand{\ds}{\displaystyle}
\newcommand{\la}{\label}

\newcommand{\beqn}{\begin{eqnarray}}
\newcommand{\eeqn}{\end{eqnarray}}
\newcommand{\be}{\begin{equation}}
\newcommand{\ee}{\end{equation}}
\newcommand{\ba}{\begin{array}}
 \newcommand{\ea}{\end{array}}
\newcommand{\pa}{\partial}
\newcommand{\fr}{\frac}

\newcommand{\ci}{\cite}
\newcommand{\re}{\ref}
\renewcommand{\Re}{\mathop{\mathrm{Re}}} 
\renewcommand{\Im}{\mathop{\mathrm{Im}}}

\DeclareMathOperator{\supp}{supp}

\begin{document}

\renewcommand{\theequation}{\thesection.\arabic{equation}}
\newtheorem{theorem}{Theorem}[section]
\renewcommand{\thetheorem}{\arabic{section}.\arabic{theorem}}
\newtheorem{definition}[theorem]{Definition}
\newtheorem{deflem}[theorem]{Definition and Lemma}
\newtheorem{lemma}[theorem]{Lemma}
\newtheorem{example}[theorem]{Example}
\newtheorem{remark}[theorem]{Remark}
\newtheorem{remarks}[theorem]{Remarks}
\newtheorem{cor}[theorem]{Corollary}
\newtheorem{prop}[theorem]{Proposition}
\pagestyle{myheadings}



\title{\footnotesize{\textit{}}\hfill\hbox{}
\newline\hbox{}\newline\Large{}}
\newcommand\address[1]{\textit{#1}}

\begin{titlepage}

\begin{center}
{\Large\bf
Scattering of Solitons for Dirac \medskip\\
Equation Coupled to a Particle}\\
\vspace{1cm}
{\large A.~I.~Komech}
\footnote{
Supported partly
 by the Alexander von Humboldt
Research Award.
}
\\
{\it Fakult\"at f\"ur Mathematik, Universit\"at Wien\\
and Institute for Information Transmission Problems RAS}\\
 e-mail:~alexander.komech@univie.ac.at
\medskip\\
{\large E.~A.~Kopylova}
\footnote{Supported partly by  the grants of FWF,
DFG and RFBR.}
\\
{\it Institute for Information Transmission Problems RAS\\
B.Karetnyi 19, Moscow 101447,GSP-4, Russia}\\
e-mail:~elena.kopylova@univie.ac.at
\medskip\\
{\large H.~Spohn}
\\
{\it Faculty of Mathematics, TU Munich, Garching, 85747 Germany
}\\
email: spohn@ma.tum.de
\end{center}

\vspace{2cm}
\begin{abstract}
We establish soliton-like asymptotics for finite energy solutions
to the Dirac equation coupled to a relativistic  particle.
Any solution with initial state close to the solitary manifold,
converges in long time limit to a sum of traveling wave and outgoing 
free wave. The convergence holds in global energy norm. The proof uses 
spectral theory and the symplectic projection onto solitary manifold 
in the Hilbert phase space.

\end{abstract}
\end{titlepage}
 

\setcounter{equation}{0}

\section{Introduction}

We prove the long time convergence to the sum of a soliton and dispersive wave
for the Dirac equation coupled to a relativistic particle.
The convergence holds in global energy norm for finite energy solution
with initial state close to the solitary manifold.
Our main motivation is to develop the techniques of Buslaev and Perelman 
\ci{BP1,BP2} in the context of the Dirac equation.
The development is not straightforward because of known 
peculiarities of the Dirac equation: nonpositivity of 
the energy, algebra of the Dirac matrices, etc.
We expect that the result might be extended to nonlinear relativistic 
Dirac equation relying on an appropriate development of our techniques.

Let  $\psi(x)\in\C^4$ be  a Dirac spinor field  in $\R^3$, 
coupled to a relativistic
particle with position $q$ and momentum $p$, governed by
\begin{equation} \label{D}
 \left\{\begin{array}{l}
     i\dot\psi(x,t)=[-i\al_1\pa_1-i\al_2\pa_2-i\al_3\pa_3+\beta m]\psi(x,t)+\rho(x-q(t))
  \\\\
      \dot q(t)=p(t)/\sqrt{1+p^2(t)},
  ~~~~~
      \dot p(t)=\Re\langle\psi(x,t), \nabla\rho(x-q(t))
\rangle
\end{array}\right|\; x\in\R^3
\end{equation}
where $\rho\in C(\R^3,\C^4)$ and $\langle\cdot,\cdot\rangle$ 
stands for the Hermitian
scalar product in $L^2(\R^3)\otimes\C^4$. Here  $\pa_j=\ds\pa/\pa x_j$,
$\al_j$ and $\beta$ are $4\times 4$ Dirac matrices.
The standard representation for the Dirac matrices $\al_j$ and  $\beta$ (in $2\times 2$ blocks) is
\be\la{ba}
\beta=\al_0= \left(\ba{cc}
I_2 & 0\\
0 & -I_2\\
\ea  \right),\quad
\al_j=
\left(\ba{ll}
0 & \sigma_j\\
\sigma_j & 0\\
\ea  \right),\quad j=1,2,3
\ee
where  $I_2$ denotes the  unit  $2\times 2$ matrix and
$$
\sigma_1=  \left(\ba{ll}
0 & 1\\
1 & 0\\
\ea         \right),\quad
\sigma_2=   \left(\ba{cc}
0 & -i\\
i & 0\\
\ea         \right),\quad
\sigma_3=   \left(\ba{cc}
1 & 0\\
0 & -1\\
\ea         \right)
$$
The matrices $\al_j$, $j=0,1,2,3$ are Hermitian  and
satisfy the anticommutation relations
\be\la{aij}
\al_j^*=\al_j,\quad\al_j\al_k+\al_k\al_j=2\delta_{jk}
\ee
We will use the following 
{\it real orthogonality relations}  
\be\la{al-pr}
\beta \psi\cdot\al_j\psi=0,~~j=1,3,~~{\rm and}~~\al_2 \psi\cdot\psi
=0,
~~~~~~\psi\in\R^4
\ee
The system \eqref{D} is translation-invariant and admits soliton solutions
\be\la{sosol}
 s_{a,v}(t)=(\psi_{v}(x-vt-a), vt+a,p_v), \quad p_v=v/\sqrt{1-v^2}
\ee
for all $a,v\in\R^3$ with $|v|<1$.
The states $S_{a,v}:=s_{a,v}(0)$ form the solitary manifold
\be\la{soman}
{\cal S}:=\{ S_{a,v}: a,v\in\R^3, |v|<1 \}
\ee
Our main result is the soliton-type asymptotics 
\be\la{Solas}
  \psi(x,t)\sim\psi_{v_\pm }(x-v_\pm t-a_\pm)+W_0(t)\phi_\pm,\quad t\to\pm\infty
\ee
for solutions to  (\re{D}) with initial data close to the solitary manifold ${\cal S}$.
Here $W_0(t)$ is the dynamical group of
the free Dirac equation, $\phi_\pm$ are the corresponding
{\it asymptotic scattering states}, and the asymptotics hold
{\it in the global} norm of the Hilbert space 
$L^2(\R^3)\otimes\C^4$.
For the particle trajectory we prove that
\be\la{qq}
  \dot q(t)\to v_\pm,\quad q(t) \sim v_\pm t+a_\pm,\quad t\to\pm\infty
\ee
The results are established under the following conditions
on the complex valued charge distribution: for  some $\nu>5/2$
\be\la{ro}
(1+|x|)^\nu|\pa^{\alpha}\rho|\in L^2(\R^3),\quad |\alpha|\le 3
\ee
We assume $\rho(-x)=\rho(x)$, $x\in\R^3$,
 for the simplicity of calculations.
Finally, we assume the  Wiener condition for the Fourier
transform $\hat\rho=(2\pi)^{-3/2}\ds\int e^{ikx}\rho(x)dx$
\be\la{W}
 {\cal B}(k)=m\beta\hat\rho(k)\cdot\hat\rho(k)>0,\quad k\in\R^3
\ee
which is the nonlinear version of the 
Fermi Golden Rule in our case
(cf. \ci{BS,Sig,SW3,SW4}):
the nonlinear perturbation is not orthogonal to the eigenfunctions
of the continuous spectrum of the linear part.
The examples are easily constructed. Namely, let us  rewrite 
(\re{W}) in the form
\be\la{Wr}
 {\cal B}(k)=m[|\hat\rho_1(k)|^2+|\hat\rho_2(k)|^2-
|\hat\rho_3(k)|^2-|\hat\rho_4(k)|^2]>0,\quad k\in\R^3
\ee
Therefore, we can take e.g.  $\rho_1$ constructed in \ci{KSK},
and $\rho_2=\rho_3=\rho_4=0$.

The system (\ref{D}) describes the charged particle interacting with
its ``own'' Dirac field. The asymptotics (\ref{Solas})-(\ref{qq})
mean asymptotic stability of uniform motion, i.e. ``the law of inertia''.
The stability is caused by ``radiative damping'', i.e. radiation of
energy to infinity appearing analytically as a local energy decay for 
solutions to the linearized equation.
The radiative damping
was suggested first by M.Abraham  in 1905 
 in the
context of the
Maxwell-Lorentz equations, \cite{A}.

One could also expect asymptotics (\ref{Solas}) for small perturbations of
the solitons for
the relativistic nonlinear Dirac equations and for the coupled
nonlinear
Maxwell-Dirac equations whose  solitons were constructed in
\ci{EGS}.
Our result models  this situation though the relativistic 
case is still open problem.

Asymptotics of type (\ref{Solas})-(\ref{qq}) were obtained previously
for the  Klein-Gordon and Schr\"odinger  equations coupled 
to the particle \cite{ikv,KK}.
More weak asymptotics of type (\ref{Solas}) 
in the {\it local energy norms}, and
without the dispersive wave, were 
obtained in \ci{IKM} and \ci{KS} for 
the Maxwell-Lorentz and wave equations respectively.

Let us comment on our approach.
For 1D translation invariant Schr\"odinger equation, 
asymptotics of type
(\ref{Solas}) were proved for the first time by Buslaev and Perelman
\cite{BP1,BP2,BS}, and extended by Cuccagna \cite{Cu} for higher dimensions.
Here we develop 
the approach \cite{ikv} 
 where
the general  Buslaev and Perelman strategy has been developed 
for the case of the Klein-Gordon equation: 
i) symplectic orthogonal decomposition of the dynamics near the
solitary manifold,  
ii) modulation equations for the symplectic projection onto the 
manifold, and iii) the 
time decay in the
transversal directions, etc
(see more details in Introduction \ci{ikv}).
We prove the asymptotics 
(\re{Solas})-- (\re{qq}) in Sections 3-11 developing this general
strategy.
One of  difficulties  is caused by well known nonpositivity of the
Hamiltonian for the Dirac equation. Respectively, the energy
conservation does not provide a priori estimate 
for the solution. We obtain  linear in time estimate for $L^2$ norm
of the solution using unitarity of the free Dirac propagator.
The main novelty in our case is thorough 
establishing the appropriate
decay of the linearized dynamics 
in Sections 12-17, and Appendices A, B, and C:
\medskip

I. Main difficulty 
 lies in the proof of the decay $\sim t^{-3/2}$ 
in weighted norms for the free Dirac equation. Here we prove the decay
for the first time (Lemma \re{Vain}). 
The proof relies on the ``soft version'' of the 
strong Huygens principle for the Dirac equation. 
Namely, the free Dirac 
propagator is concentrated mainly near the light cone, while the
contribution of the inner zone is a Hilbert-Schmidt operator.

II. Next difficulty lies in the computation of 
the  spectral properties of the linearized equation at the soliton.
We 
do not postulate any spectral properties of the equation
in contrast to 
majority of the 
 works in the field.
Namely, we
find that under the Wiener condition (\re{W}), the discrete
spectrum consists only from zero point with algebraic multiplicity  6
(Lemma \re{vdec}).
The multiplicity is totally due to the translation invariance of the 
system (\re{D}). 

III. Moreover, we exactly calculate the symplectic orthogonality 
conditions
(\re{soin})
for initial data of the linearized equation.
These conditions are necessary for the proof of the decay.

IV. All computations differ significantly from the case of the
Klein-Gordon  equation \cite{ikv} because of the algebra of the Dirac
matrices. An important role play the 
{\it real orthogonality relations}
(\re{al-pr}) for the Dirac matrices.

Our paper is organized as follows. In Section 2, we formulate the main
result. In  Section 3, we introduce the symplectic projection onto the
solitary manifold. The linearized equation is considered in Sections 4
and 5. In Section 6, we split the dynamics in two components:
along the solitary manifold, and in transversal directions. The time
decay of the transversal component is established in sections 7-10
using the time decay of the linearized dynamics. In
Section 11 we prove the main result. In Sections 12 - 16 we
justify
the time decay of the linearized dynamics relying on the 
 weighted  decay for the
free Dirac equation in a moving frame which is proved 
in  Section 17.
In Appendices A,  B  and C we collect some technical calculations.
\bigskip\\
\setcounter{equation}{0}
\section{Main results}\la{mr}
\subsection{Existence of dynamics}
We consider the Cauchy problem for the  system  (\re{DD}) which we write as
\be\la{SD}
  \dot Y(t)=F(Y(t)),\quad t\in\R:\quad Y(0)=Y_0
\ee
Here $Y(t)=(\psi(t), q(t), p(t))$, $Y_0=(\psi(0), q_0, p_0)$, and all
derivatives are understood in the sense of distributions.
To formulate our results precisely, we need some definitions. We introduce
a suitable phase space for equation  (\ref{SD}).
Let $L^2_{\alpha}$, $\alpha\in\R$, denote the   weighted Agmon spaces with the norm 
$\Vert\psi\Vert_{\alpha}=\Vert\psi\Vert_{L^2_\alpha}:=\Vert(1+|x|)^{\alpha}|\psi|\Vert_{L^2}$,
where $L^2=L^2(\R^3)$.
\begin{definition}
  i) The phase space ${\cal E}$ is the   Hilbert space 
$L^2_0\oplus {\R}^3\oplus {\R}^3$
  of states  $Y=(\psi,q,p)$ with the finite norm
  $$
  \Vert Y\Vert_{\cal E}=\Vert\psi\Vert_{0}+|q|+|p|
  $$
  ii) ${\cal E}_{\alpha}$ is the space 
  $L^2_{\alpha}\oplus {\R}^3\oplus {\R}^3$
  with the  finite norm
  $$
  \Vert Y\Vert_{\alpha}=\Vert \,Y\Vert_{{\cal E}_{\alpha}}=
  \Vert\psi\Vert_{\alpha}+|q|+|p|
$$
\end{definition}

\begin{prop}\la{WPexistence}
  Let (\re{ro}) hold. Then\\
  (i) For every $Y_0\in {\cal E}$ the Cauchy problem (\re{DD}) has a unique
  solution $Y(t)\in C(\R, {\cal E})$.\\
  (ii) For every $t\in\R$, the map $U(t): Y_0\mapsto Y(t)$ is continuous on ${\cal E}$.\\
\end{prop}
\begin{proof}
{\it Step i)}
First, let us fix an arbitrary $b>0$ and prove {\it (i)-(ii)} for $Y_0\in{\cal E}$ such that
$\Vert\psi_0\Vert_{0}\le b$ and $|t|\le\ve=\ve(b)$ for some sufficiently small $\ve(b)>0$.
Let us rewrite the Cauchy problem \eqref{SD} as
\be\la{WP2.2}
  \dot Y(t)=F_1(Y(t))+F_2(Y(t)),\quad t\in\R:\quad Y(0)=Y_0
\ee
where 
$$
F_1: Y\mapsto ((-\al_j\pa_j-i\beta m)\psi,0,0)
$$
$$
F_2: Y\mapsto (-i\rho(x-q), \,p/\sqrt{1+p^2},\, \Re\int\psi\cdot\na\rho(x-q)dx)
$$
The Fourier transform provides the existence and uniqueness of solution
$Y_1(t)\in C(\R,{\cal E})$ to the linear problem \eqref{WP2.2} with $F_2=0$.
Let $U_1(t): Y_0\mapsto Y_1(t)$ be the corresponding strongly continuous
group of bounded linear operators on ${\cal E}$. Then \eqref{WP2.2} for
$Y(t)\in C(\R,{\cal E})$ is equivalent to the integral Duhamel equation
\be\la{Dugamel}
  Y(t)=U_1(t)Y_0+\int\limits_0^t~ds~U_1(t-s)F_2(Y(s))
\ee
because $F_2(Y(\cdot))\in  C(\R,{\cal E})$ in this case.
The latter follows from  local Lipschitz continuity of the map $F_2$ in ${\cal E}$:
for each $b>0$ there exist a $\varkappa=\varkappa(b)>0$ such that for
all $Y_1=(\psi_1,q_1,p_1), Y_2=(\psi_2,q_2,p_2)\in{\cal E}$
with $\Vert \psi_1\Vert_{0},\;\Vert\psi_2\Vert_{0}\le b$,
$$
  \Vert F_2(Y_1)-F_2(Y_2)\Vert_{\cal E}\le\varkappa\Vert Y_1-Y_2\Vert_{\cal E}
$$
Therefore, by the contraction mapping principle, equation \eqref{Dugamel}
has a unique local solution $Y(\cdot)\in C([-\ve,\ve],{\cal E})$ with
$\ve>0$ depending only on $b$.\\
{\it Step ii)}
Second we derive {\it a priori} estimate. Consider 
$\psi_0\in C_0^{\infty}(\R^3)\otimes\C^4$.
Then
$$
\fr{d}{dt}\Vert\psi\Vert_{0}^2=\int(\ov\psi\cdot\dot\psi+\psi\cdot\dot{\ov\psi})dx
=\int(i\ov\psi\cdot\rho(x-q)-i\psi\cdot{\ov\rho(x-q)})dx \le C \Vert\psi\Vert_{0}
$$
Hence,
$$
\Vert\psi(t)\Vert_{0}\le \fr 12Ct+\Vert\psi(0)\Vert_{0}
$$
Now, the last two equalities (\ref{D}) imply  {\it a priori}
estimates for $|\dot p|$ and $|\dot q|$. The {\it a priori} estimates
for general  initial data 
$\psi_0\in L^2_0$ follow by approximating  initial data by the
functions from $C_0^{\infty}(\R^3)\otimes\C^4$. \\
{\it Step iii)} Properties {\it (i)-(ii)} for arbitrary $t\in\R$ now
follow from the same properties for small $|t|$ and from  {\it a priori} estimate.
\end{proof}
\subsection{Solitary manifold and main result}
Let us compute the  solitons (\re{sosol}). The substitution to (\re{D}) gives the 
stationary equations
\be\la{stfch}
  \left.\begin{array}{l}
  -iv_j\pa_j\psi_v(y)=[-i\al_j\pa_j+\beta m]\psi_v(y)+\rho(y)
  \\\\
  v=p_v/\sqrt{1+p_v^2},\quad
  0=\Re\ds\int\psi_v(y)\cdot\na\rho(y)\,dy
  \end{array}\right|
\ee
Applying Fourier transform to the first equation in (\ref{stfch}) we obtain
$$
(-v_jk_j +\al_jk_j-\beta m)\hat\psi_v(k)=\hat\rho(k)
$$
hence
\be\la{hp}
\hat\psi_v(k)=-\frac{(v_jk_j +\al_jk_j-\beta m)\hat\rho(k)}
{(v_jk_j +\al_jk_j-\beta m)(v_jk_j -\al_jk_j+\beta m)}
=\frac{(v_jk_j +\al_jk_j-\beta m)\hat\rho(k)}{k^2+ m^2-(v_jk_j)^2}
\ee
The soliton is given by the formula
\be\la{sol}
\psi_v(x) = \ds
 \frac {i\gamma} {4\pi} (v_j\pa_j+\al_j\pa_j+i\beta m)\int
\frac {e^{-m|\gamma(y-x)_\|+(y-x)_\bot|}\rho(y)d^3y}
{|\gamma(y-x)_\|+(y-x)_\bot|},\quad
p_v= \gamma v=\ds\fr v{\sqrt{1-v^2}}
\ee
It remains to prove that the last equation of (\re{stfch}) holds. Indeed,
 Parseval identity and equality (\ref{hp}) imply
$$
\Re\!\!\int\!\!\psi_v(y)\cdot \pa_j\rho(y)dy=\Re\!\!\int\!\! ik_j\hat\psi_v(k)\cdot\hat\rho(k)dk
=\Re\!\!\int\!\! ik_j\frac{(v_jk_j \!+\!\al_jk_j\!-\!\beta m)\hat\rho(k)\cdot\hat\rho(k)}
{k^2+ m^2-(v_jk_j)^2}dk=0
$$
since the integrand is pure imaginary   function.
Hence, the soliton solution (\re{sosol})
exists and is defined uniquely for any couple $(a,v)$ with $|v|<1$ and
$a\in\R^3$. Let us denote by  $V:=\{v\in\R^3:|v|<1\}$.
\begin{definition}
  A soliton state is $S(\si):=(\psi_{v}(x-b),b,v)$, where
  $\si:=(b,v)$ with $b\in\R^3$ and $v\in V$.
\end{definition}
Obviously, the soliton solution admits the representation $S(\si(t))$, where
\be\la{sigma}
  \si(t)=(b(t),v(t))=(vt+a,v)
\ee
\begin{definition}
  A solitary manifold is the set ${\cal S}:=\{S(\si):\si\in\Sigma:=\R^3\times V\}$.
\end{definition}

The main result of our paper is the following theorem.
\begin{theorem}\la{main}
  Let \eqref{ro}, and the Wiener condition (\re{W}) hold. Let $\nu>5/2$ be the number from
  \eqref{ro}, and $Y(t)$ be the solution to the Cauchy problem  (\re{SD}) with
  the initial state $Y_0$ which is sufficiently close to the solitary manifold:
  \be\la{close}
   d_0:={\rm dist}_{{\cal E}_\nu}(Y_0,{\cal S})\ll 1
  \ee
  Then the asymptotics  hold for $t\to\pm\infty$,
  \be\la{nas}
    \dot q(t)=v_{\pm}+{\cal O}(|t|^{-2}),\quad q(t)=v_{\pm}t+a_{\pm}+{\cal O}(|t|^{-1})
  \ee
  \be\la{nasf}
    \psi(x,t)=\psi_{v\pm}(x-v_{\pm}t-a_{\pm})+W_0(t)\phi_{\pm}+r_{\pm}(x,t)
  \ee
  with
  \be\la{rm}
     \Vert r_\pm(t)\Vert_{0}={\cal O}(|t|^{-1/2})
  \ee
\end{theorem}
It suffices to  prove the asymptotics  (\re{nas}), (\re{nasf}) for $t\to+\infty$
since the system (\re{D}) is time reversible.

\setcounter{equation}{0}
\section{Symplectic projection}
\subsection{Hamiltonian  structure}
Denote  $\psi_1=\Re\psi$, $\psi_2=\Im\psi$, $\rho_1=\Re\rho$, $\rho_2=\Im\rho$,
$\tilde\al_2=-i\al_2$. Then the system \eqref{D} reads
\begin{equation} \label{DD}
 \!\!\!\!\!\left\{\begin{array}{l}
\dot\psi_1(x,t)=-(\al_1\pa_1+\al_3\pa_3)\psi_1(x,t)+(\tilde\al_2\pa_2+\beta m)\psi_2(x,t)
+ \rho_2(x-q(t))\\
\dot\psi_2(x,t)=-(\tilde\al_2\pa_2+\beta m)\psi_1(x,t)-(\al_1\pa_1+\al_3\pa_3)\psi_2(x,t)
 - \rho_1(x-q(t)) \\
\dot q(t)=p(t)/\sqrt{1+p^2(t)}\\
\dot p(t)=\ds\int\bigl(\psi_1(x,t)\cdot\nabla\rho_1(x-q(t))
+\psi_2(x,t)\cdot\nabla\rho_2(x-q(t))\bigr)dx
\end{array}\right|\; x\in\R^3
\end{equation}
This is a Hamilton system with the Hamilton functional
\beqn\nonumber
  {\cal
    H}(\psi_1,\psi_2,q,p)\!\!\!&=&\!\!\!\frac12\!\int(\psi_1\cdot(\tilde\al_2\pa_2
\!+\!\beta m)\psi_1+\psi_2\cdot(\tilde\al_2\pa_2\!+\!\beta m)\psi_2
  +2\psi_1\cdot(\al_1\pa_1\!+\!\al_3\pa_3)\psi_2)dx\\
\label{Hamil}
\!\!\!&+&\!\!\!\int(\psi_1(x)\cdot\rho_1(x-q)+\psi_2(x)\cdot\rho_2(x-q))dx+\sqrt{1+p^2}
\eeqn
Equation (\ref{DD})  can be written as a Hamilton system
\be\la{ham}
\dot Y=JD{\cal H}(Y),\quad Y=(\psi_1,\psi_2,q,p),\quad J:=\left(
  \begin{array}{cccc}
  0 & I_4 & 0 & 0\\
  -I_4 & 0 & 0 & 0\\
  0 & 0 & 0 & I_3\\
  0 & 0 & -I_3 & 0\\
  \end{array}
  \right)
\ee
where $D{\cal H}$ is the Fr\'echet derivative with respect to $\psi_{1k}$, $\psi_{2k}$,
$k=1,2,3,4$, $p$ and $q$ of the Hamilton functional.

\subsection{Symplectic projection onto solitary manifold}
Let us identify the tangent space to ${\cal E}$, at every point, with ${\cal E}$.
Consider the symplectic form $\Om$ defined on
${\cal E}$ by $\Om=\ds\int d\psi_1(x)\wedge d\psi_2(x)\,dx+dq\wedge dp$, i.e.
\be\la{OmJ}
\Om(Y^1,Y^2)=\langle Y^1,JY^2\rangle,\quad Y^j=(\psi_1^j,\psi_2^j,q^j,p^j)\in {\cal E},
\quad j=1,2
\ee
where
$$
  \langle Y^1,Y^2\rangle:=\langle\psi_{1}^1,\psi_{1}^2\rangle+
  \langle\psi_{2}^1,\psi_{2}^2\rangle+q^1\cdot q^2+p^1\cdot  p^2
$$
and
$\langle\psi_{1}^1,\psi_{1}^2\rangle=\ds\int\psi_{1}^1(x)\cdot\psi_{1}^2(x)dx$ 
stands for the scalar product or its different extensions.
It is clear that the form $\Om$ is non-degenerate, i.e.
$$
  \Om(Y^1,Y^2)=0\,\,~\mbox{\rm for every}~~ \,Y^2\in{\cal E} \,\,\Longrightarrow\,\, Y^1=0
$$
\begin{definition}
  i) $Y^1\nmid Y^2$ means that $Y^1\in{\cal E}$, $Y^2\in{\cal E}$,
  and $Y^1$ is symplectic orthogonal to $Y^2$, i.e. $\Om(Y^1,Y^2)=0$.

  ii) A projection operator ${\bf P}:{\cal E}\to{\cal E}$ is called symplectic orthogonal if
  $Y^1\nmid Y^2$ for $Y^1\in\mbox{\rm Ker}\,{\bf P}$ and $Y^2\in \Im{\bf P}$.
\end{definition}
Let us consider the tangent space ${\cal T}_{S(\si)}{\cal S}$ to
the manifold ${\cal S}$ at a point $S(\si)$.
The vectors $\tau_j:=\pa_{\si_j}S(\si)$, where
$\pa_{\si_j}:=\pa_{b_j}$ and $\pa_{\si_{j+3}}:=\pa_{v_{j}}$ with $j=1,2,3$,
form a basis in ${\cal T}_{\si}{\cal S}$. In detail,
\be\la{inb}
\left.\begin{array}{rclrrrrcrcl}
  \tau_j=\tau_j(v)&:=&\pa_{b_j}S(\si)=
  (&\!\!\!\!-\pa_j\psi_{v1}(y)&\!\!\!\!,&\!\!\!\!-\pa_j\psi_{v2}(y)
   &\!\!\!\!,&\!\!e_j&\!\!\!\!,&\!\!0&\!\!\!\!)\\
   \tau_{j+3}=\tau_{j+3}(v)&:=&\pa_{v_j}S(\si)=
   (&\!\!\!\!\pa_{v_j}\psi_{v1}(y)&\!\!\!\!,&\!\!\!\!
   \pa_{v_j}\psi_{v2}(y)&\!\!\!\!,&\!\!0&\!\!\!\!,&\!\!\pa_{v_j}p_v&\!\!\!\!)
\end{array}\right|\quad j=1,2,3
\ee
where $\psi_{v1}=\Re\psi_v$,  $\psi_{v2}=\Im\psi_v$, 
$y:=x-b$ is the ``moving frame coordinate'', $e_1=(1,0,0)$ etc. Let us stress that
the functions $\tau_j$ will be considered always as the functions of $y$, not of $x$.
Formula (\re{sol}) and  condition (\re{ro}) imply that
\be\la{tana}
   \tau_j(v)\in{\cal E}_\al,\quad v\in V,\quad j=1,\dots,6,\quad \forall\al\in\R
\ee
\begin{lemma}\la{Ome}
  The matrix with the elements $\Om(\tau_l(v),\tau_j(v))$ is non-degenerate
   $\forall v\in V$.
\end{lemma}
\begin{proof}
The elements are computed in Appendix A.
  As the result, the matrix $\Om(\tau_l,\tau_j)$ has the form
  \be\la{Omega}
     \Om(v):=(\Om(\tau_l,\tau_j))_{l,j=1,\dots,6}=\left(
     \begin{array}{ll}
         0 & \Om^+(v)\\
         -\Om^+(v) & 0
     \end{array}\right)
  \ee
  where the $3\times3$-matrix $\Om^+(v)$ equals
  \be\la{Wm}
     \Om^+(v)=K+(1-v^2)^{-1/2}E+(1-v^2)^{-3/2}v\otimes v
  \ee
  Here $K$ is a symmetric $3\times3$-matrix with the elements \\
 \be\la{alpha}
 K_{ij}=\int dkk_jk_l{\cal B}(k)\frac{k^2+m^2+3(v_jk_j)^2}{(k^2+m^2-(v_jk_j)^2)^3}\quad
\ee
  where ${\cal B}(k)>0$ is defined in  (\re{W}).
  The matrix $K$ is the integral of the symmetric nonnegative definite
  matrix $k\otimes k=(k_ik_j)$ with a positive weight.
  Hence, the matrix $K$ is  nonnegative definite.
  Since the unite matrix $E$ is positive definite, the matrix $\Om^+(v)$ is  symmetric and
  positive definite, hence non-degenerate. Then
  the matrix $\Om(\tau_l,\tau_j)$ also is non-degenerate.
\end{proof}
Let us introduce the translations $T_a:(\psi(\cdot),q,p)\mapsto(\psi(\cdot-a),q+a,p)$, 
$a\in\R^3$.
Note that the manifold ${\cal S}$ is invariant with respect to the translations.
Let us denote $v(p):=p/\sqrt{1+p^2}$ for $p\in\R^3$.
\begin{definition}
   i) For any $\al\in\R$ and $\overline v<1$ denote by
   ${\cal E}_\al(\overline v)=\{Y=(\psi,q,p)\in{\cal E}_\al:|v(p)|\le\overline v\}$.
   We set ${\cal E}(\overline v):={\cal E}_0(\overline v)$.
\\
  ii) For any  $\tilde v<1$ denote by
  $\Si(\tilde v)=\{\si=(b,v):b\in\R^3, |v|\le \tilde v \}$.
\end{definition}
The next Lemma provide that in a small neighborhood of the soliton
manifold ${\cal S}$ a ``symplectic orthogonal projection''
onto ${\cal S}$ is well-defined.
\begin{lemma}\la{skewpro}(cf.\cite[ Lemma 3.4]{ikv})
  Let (\re{ro}) hold, $\al\in\R$ and $\ov v<1$. Then
  \\
  i) there exists a neighborhood ${\cal O}_{\al}({\cal S})$ of ${\cal S}$ in ${\cal E}_\al$
  and a map ${\bf\Pi}:{\cal O}_{\al}({\cal S})\to{\cal S}$ such that ${\bf\Pi}$ is uniformly
  continuous on ${\cal O}_{\al}({\cal S})\cap {\cal E}_\al(\overline v)$ 
in the metric of ${\cal E}_\al$,
\be\la{proj}
     {\bf\Pi} Y=Y~~\mbox{for}~~ Y\in{\cal S}, ~~~~~\mbox{and}~~~~~
     Y-S \nmid {\cal T}_S{\cal S},~~\mbox{where}~~S={\bf\Pi} Y
  \ee
  ii) ${\cal O}_{\al}({\cal S})$ is invariant with respect to the translations $T_a$, and
  \[ {\bf\Pi} T_aY=T_a{\bf\Pi} Y,~~~~~\mbox{for}~~Y\in{\cal O}_{\al}({\cal S})
  ~~\mbox{and}~~a\in\R^3 \]
  iii) For any $\ov v<1$ there exists a $\tilde v<1$ s.t. ${\bf\Pi} Y=S(\si)$
  with  $\si\in \Si(\tilde v)$ for
  $Y\in{\cal O}_{\al}({\cal S})\cap {\cal E}_\al(\overline v)$.
  \\
  iv) For any $\tilde v< 1$ there exists an $r_\al(\tilde v)>0$ s.t.
  $S(\si)+Z\in{\cal O}_{\al}({\cal S})$ if $\si\in\Si(\tilde v)$
  and $\Vert Z\Vert_\al<r_\al(\tilde v)$.
\end{lemma}
\smallskip
\noindent
We will call ${\bf\Pi}$ a symplectic orthogonal projection onto ${\cal S}$.
\begin{cor}
  The condition (\re{close}) implies that $Y_0=S+Z_0$ where
  $S=S(\si_0)={\bf\Pi} Y_0$, and
\be\la{closeZ}
 \Vert Z_0\Vert_\nu \ll 1
\ee
  \end{cor}
\begin{proof}
   Lemma \re{skewpro} implies that ${\bf\Pi} Y_0=S$ is well defined for small $d_0>0$.
   Furthermore, the condition (\re{close}) means that there exists a point $S_1\in{\cal S}$
   such that $\Vert Y_0-S_1\Vert_\nu=d_0$. Hence,
   $Y_0,S_1\in {\cal O}_{\nu}({\cal S})\cap{\cal E}_\nu(\overline v)$
   with a $\overline v< 1$ which does not depend on $d_0$ for sufficiently small $d_0$.
   On the other hand, ${\bf\Pi} S_1= S_1$, hence the uniform continuity of the map ${\bf\Pi}$
   implies that $\Vert S_1- S\Vert_\nu\to 0$ as $d_0\to 0$. Therefore, finally,
   $\Vert Z_0\Vert_\nu=\Vert Y_0- S \Vert_\nu \le \Vert Y_0- S_1 \Vert_\nu+
   \Vert S_1-S  \Vert_\nu\le d_0+o(1)\ll 1$ for small $d_0$.
\end{proof}
\setcounter{equation}{0}
\section{Linearization on  solitary manifold}
Let us consider a solution to the system (\re{DD}), and split it as the sum
\be\la{dec}
   Y(t)=S(\si(t))+Z(t)
\ee
where $\si(t)=(b(t),v(t))\in\Sigma$ is an arbitrary smooth function of $t\in\R$.
In detail, denote $Y=(\psi,q,p)$ and $Z=(\Psi,Q,P)$.
Then (\re{dec}) means that
\be \la{add}
  \psi(x,t)=\psi_{v(t)}(x-b(t))+\Psi(x-b(t),t),~~q(t)=b(t)+Q(t),
  ~~p(t)=p_{v(t)}+P(t)
\ee
Let us substitute (\re{add}) to (\re{D}), and linearize the equations in $Z$.
Setting $y=x-b(t)$ which is the ``moving frame coordinate'', we obtain  that
\be\la{addeq}\left.
  \begin{array}{rcl}
    \dot\psi&=&\dot v\cdot \na_v\psi_{v}(y)-\dot b\cdot \na\psi_{v}(y)+
    \dot\Psi(y,t)-\dot b\cdot \na\Psi(y,t)\\\\
    &=&[-\al_j\pa_j-i\beta m](\psi_{v}(y)+\Psi(y,t))-i\rho(y-Q)\\\\
    \dot q&=&\dot b+\dot Q=\ds\fr{p_v+P}{\sqrt{1+(p_v+P)^2}}\\\\
    \dot p&=&\dot v\cdot\na_vp_v+\dot P=\Re\langle\psi_{v}(y)+
    \Psi(y,t),\na\rho(y-Q)\rangle
\end{array}\right|
\ee
Let us extract linear terms in $Q$. First note that
$\rho(y-Q)=\rho(y)-Q \cdot\na\rho(y)+N_1(Q)$, 
$\nabla\rho(y-Q)=\nabla\rho(y)-\nabla (Q\cdot\na\rho(y))+\tilde N_1(Q)$. 
 
The condition \eqref{ro} implies that for $N_1(Q)$ and $\tilde N_1(Q)$
the bound  holds,
\be\la{N}
\Vert N_1(Q)\Vert_{\nu}+\Vert\tilde N_1(Q)\Vert_{\nu}
\le C_\nu(\overline Q)Q^2
\ee
uniformly in $|Q|\le\overline Q$ for any fixed $\overline Q$, where $\nu$ is the parameter
from Theorem \ref{main}. Second, the Taylor expansion gives
$$
\ds\fr{p_v+P}{\sqrt{1+(p_v+P)^2}}=v+\frac 1{\gamma}(P-v(v\cdot P))+N_2(v,P)
$$
where $1/\gamma=\sqrt{1-v^2}=(1+p_v^2)^{-1/2}$, and
\be\la{N2}
|N_2(v,P)|\le C(\tilde v)P^2
\ee
uniformly with respect to $|v|\le\tilde v<1$.
Using the equations (\re{stfch}), we obtain from (\re{addeq})
the following equations for the components of the vector $Z(t)$:
\be\la{Phi}
\left.\begin{array}{rcl}
\dot \Psi(y,t)\!\!\!&= &\!\!\![-\al_j\pa_j-i\beta m]\Psi(y,t)
+\dot b\cdot \na\Psi(y,t)+iQ \cdot\na\rho(y)\\\\
\!\!\!&+&\!\!\!(\dot b-v)\cdot \na\psi_{v}(y)-\dot v\cdot \na_v\psi_{v}(y)-iN_1
\\\\
\dot Q(t)\!\!\! &=&\!\!\!\frac 1{\gamma}(E-v\otimes v)P+(v-\dot b)+N_2
\\\\
\dot P(t)\!\!\! &=&\!\!\!-\dot v\cdot\na_vp_v+\Re\langle\Psi(y,t),\na\rho(y)\rangle+
\Re\langle\na\psi_{v}(y),Q \cdot\na\rho(y)\rangle+N_3(v,Z)
\end{array}\right|
\ee
where $N_3(v,Z)=-\Re\langle\na\psi_{v},N_1(Q)\rangle-
\Re\langle\Psi,\na(Q\cdot \na\rho)\rangle+\Re\langle\Psi,\tilde N_1(Q)\rangle$.
Clearly, $N_3(v,Z)$ satisfies the following estimate
\be\la{N3}
  |N_3(v,Z)|\le C_\nu(\rho,\overline v,\overline Q)
  \Big[Q^2+\Vert\Psi\Vert_{-\nu}|Q|\Big]
\ee
uniformly in $|v|\le\tilde v$ and $|Q|\le \overline Q$
for any fixed  $\tilde v<1$. For the vector version $Z=(\Psi_1,\Psi_2,Q,P)$
with $\Psi_1=\Re\Psi$, $\Psi_2=\Im\Psi$  we rewrite the equations (\re{Phi}) as
\be\la{lin}
\dot Z(t)=A(t)Z(t)+T(t)+N(t),\,\,\,t\in\R
\ee
Here the operator $A(t)=A_{v,w}(t)$ depends on two parameters, $v=v(t)$, and
$w=\dot b(t)$ and can be written in the form
\be\la{AA}
  A_{v,w}\left(\!\!\begin{array}{c}
     \Psi_1 \\ \Psi_2 \\ Q \\ P
\end{array}\!\!\right)\!=\!\left(\!\begin{array}{cccc}
-\!\al_1\pa_1\!-\!\al_2\pa_2\!+\!w\cdot\!\na &\tilde\al_2\pa_2+\beta m & -\na\rho_2\cdot & 0 \\
-(\tilde\al_2\pa_2+\beta m) &-\!\al_1\pa_1\!-\!\al_2\pa_2\!+\!w\cdot\!\na &\na\rho_1\cdot& 0 \\
                  0  &      0          &        0          &    B_v                       \\
\langle\cdot,\na\rho_1\rangle & \langle\cdot,\na\rho_2\rangle &
\langle\na\psi_{vj},\cdot\na\rho_j\rangle & 0
\end{array}\!\right)\left(\!\!\begin{array}{c}
      \Psi_1  \\  \Psi_2  \\ Q \\ P
\end{array}\!\!\right)
\ee
where $B_v=\frac 1{\gamma}(E-v\otimes v)$.
Furthermore,   $T(t)=T_{v,w}(t)$ and $N(t)=N(t,\si,Z)$ in  (\re{lin}) stand for
\be\la{TN}
  T_{v,w}=\left(\begin{array}{c}
  (w-v)\cdot\na\psi_{v1}-\dot v\cdot\na_v\psi_{v1}\\
 (w-v)\cdot\na\psi_{v2}-\dot v\cdot\na_v\psi_{v2}\\
  v-w \\-\dot v\cdot\na_vp_v
  \end{array}\right),\quad
  N(\si,Z)=\left(\begin{array}{c}
  N_{12}(Z) \\ -N_{11}(Z) \\ N_2(v,Z) \\ N_3 (v,Z)
  \end{array}\right)
\ee
where $v=v(t)$, $w=w(t)$, $\si=\si(t)=(b(t),v(t))$, and $Z=Z(t)$.
The estimates (\re{N}), (\re{N2}) and (\re{N3}) imply that
\be\la{N14}
  \Vert N(\si,Z)\Vert_\nu\le C(\tilde v, \overline Q)
  \Vert Z\Vert_{-\nu}^2
\ee
uniformly in $\si\in\Si(\tilde v)$ and $\Vert Z\Vert_{-\nu}\le r_{-\nu}(\tilde v)$
for any fixed  $\tilde v<1$.
\begin{remark}\la{rT}
  {\rm
  i) The term $A(t)Z(t)$ in the right hand side of the equation  (\re{lin})
  is linear  in $Z(t)$, and $N(t)$ is a {\it high order term} in $Z(t)$.
  On the other hand, $T(t)$ is a zero order term which does not vanish at $Z(t)=0$
  since $S(\si(t))$ generally is not a soliton solution if (\re{sigma})
  does not hold (though $S(\si(t))$ belongs to the solitary manifold).
  \\
  ii) Formulas (\re{inb}) and (\re{TN}) imply:
  \be\la{Ttang}
     T(t)=-\sum\limits_{l=1}^3[(w-v)_l\tau_l+\dot v_l\tau_{l+3}]
  \ee
  and hence $T(t)\in {\cal T}_{S(\si(t))}{\cal S}$, $t\in\R$.}
 \end{remark}
 \setcounter{equation}{0}
\section{ Linearized equation}

Here we collect some Hamiltonian and spectral properties of the generator (\re{AA}) 
of the linearized equation. First, let us consider the linear equation
\be\la{line}
   \dot X(t)=A_{v,w}X(t),\quad t\in\R,\quad v\in V,\quad w\in \R^3
\ee
\begin{lemma} \la{haml}(cf. Lemma 5.1 \cite{ikv})
  i) For any $v\in V$ and $w\in \R^3$ the equation (\re{line}) can be written as the Hamilton system
  (cf. (\re{ham})),
  \be\la{lineh}
     \dot X(t)=
     JD{\cal H}_{v,w}(X(t)),\quad t\in\R
  \ee
 where $D{\cal H}_{v,w}$ is the Fr\'echet derivative with respect to
  $\Psi_{1k}$, $\Psi_{2k}$,
$k=1,2,3,4$, $P$ and $Q$
  of the Hamilton functional
  \begin{multline}\la{H0}
    {\cal H}_{v,w}(X)=\fr12\int(\Psi_1\cdot(\tilde\al_2\pa_2+\beta m)\Psi_1
  +\Psi_2\cdot(\tilde\al_2\pa_2+\beta m)\Psi_2
  +2\Psi_1\cdot(\al_1\pa_1+\al_3\pa_3)\Psi_2)dy\\
     +\int\rho_j(y)Q\cdot\na\Psi_j dy+
    \fr12P\cdot B_vP-\fr12\langle Q\cdot\na\psi_{vj}(y),Q\cdot\na\rho_j(y)\rangle,\quad
      X=(\Psi_1,\Psi_2,Q,P)\in {\cal E}
  \end{multline}
  \\
 ii) The skew-symmetry relation holds,
  \be\la{com}
     \Omega(A_{v,w}X_1,X_2)=-\Omega(X_1,A_{v,w}X_2),\quad X_1\in {\cal E},\quad
X_2\in H^1(\R^3)\oplus H^1(\R^3)\oplus\R^3\oplus\R^3
  \ee
  \end{lemma}
  \begin{lemma} \la{ljf}
The operator $A_{v,w}$ acts on the tangent vectors $\tau_j(v)$ to the solitary
  manifold as follows,
  \begin{multline}\la{Atanform}
     A_{v,w}[\tau_j(v)]=(w-v)\cdot\na\tau_j(v),\,\,\,A_{v,w}[\tau_{j+3}(v)]=
     (w-v)\cdot\na\tau_{j+3}(v)+\tau_j(v),\;j=1,2,3
  \end{multline}
\end{lemma}
\begin{proof}
In detail, we have to show that
$$
A_{v,w}\left(\ba{c}
-\pa_j\psi_{v1}  \\-\pa_j\psi_{v2}                           \\ e_j \\ 0
\ea\right)=\left(\ba{c}
(v-w)\cdot\na\pa_j\psi_{v1}    \\(v-w)\cdot\na\pa_j\psi_{v2} \\ 0  \\ 0
\ea\right)
$$
\be\la{Atan}
A_{v,w}\left(\ba{c}
\pa_{v_j}\psi_{v1}   \\ \pa_{v_j}\psi_{v2}   \\ 0 \\ \pa_{v_j}p_{v}
\ea\right)=\left(\ba{c}
(w-v)\cdot\na\pa_{v_j}\psi_{v1} \\(w-v)\cdot\na\pa_{v_j}\psi_{v2} \\ 0 \\ 0
\ea\right)+\left(\ba{c}
-\pa_{j}\psi_{v1}   \\ -\pa_{j}\psi_{v2}   \\e_j \\ 0
\ea\right)
\ee
Indeed, differentiate the equations (\re{stfch}) in $b_j$
 and $v_j$, and obtain that the derivatives of soliton state in
parameters satisfy the following equations,
\be\la{stinb}
\left.\ba{rclrcl}
-v\cdot\na\pa_j\psi_{v}\!\!\!\!&=&\!\!\!\![-\al\cdot\na-i\beta m]\pa_j\psi_{v}-i\pa_j\rho
\\\\
-\pa_j\psi_{v}-v\cdot\na\pa_{v_j}\psi_{v}\!\!\!\!&=&\!\!\!\![-\al\cdot\na-i\beta m]\pa_{v_j}\psi_{v}
\\\\
\pa_{v_j}p_{v}\!\!\!\!&=&\!\!\!\! e_j(1-v^2)^{-1/2}+\ds v\fr{v_j}{(1-v^2)^{3/2}}
\\\\
0\!\!\!\!&=&\!\!\!\!\langle\pa_{v_j}\psi_{v1},\na\rho_1\rangle
+\langle\pa_{v_j}\psi_{v2},\na\rho_2\rangle
\ea~~~~~~~~~~~~~~\right|
\ee
for $j=1,2,3$.
Then (\re{Atan}) follows from (\re{stinb}) by definition
of $A$ in (\re{AA})
\end{proof}
\begin{cor}\la{ceig}
  Let  $w=v\in V$. Then $\tau_j(v)$  are eigenvectors,
  and $\tau_{j+3}(v)$ are root vectors of the
  operator $A_{v,v}$, corresponding to zero eigenvalue, i.e.
  \be\la{Atanformv}
    A_{v,v}[\tau_j(v)]=0,\,\,\,A_{v,v}[\tau_{j+3}(v)]=
    \tau_j(v),\,\,\,j=1,2,3.
  \ee
 \end{cor}
\begin{remark}
{\rm
  For a soliton solution of the system(\re{DD}) we have $\dot b=v$, $\dot v=0$,
  and hence $T(t)\equiv 0$. Thus, the equation (\re{line}) is the linearization
  of the system (\re{DD}) on a soliton solution. In fact, we do not linearize
  (\re{DD}) on a soliton solution, but on a trajectory $S(\si(t))$ with $\si(t)$
  being nonlinear in $t$. We will show later that $T(t)$ is quadratic in $Z(t)$
  if we choose $S(\si(t))$ to be  the symplectic orthogonal projection of $Y(t)$.
  Then (\re{line})  is again the linearization of (\re{DD}).
}
\end{remark}
\setcounter{equation}{0}
\section{Symplectic decomposition of  dynamics}
Here we decompose the dynamics in two components: along the manifold ${\cal S}$ and in
transversal directions. The equation (\re{lin}) is obtained without any assumption on $\si(t)$
in (\re{dec}). We are going to choose $S(\si(t)):={\bf\Pi} Y(t)$, but then we need to know that
\be\la{YtO}
  Y(t)\in {\cal O}_{-\nu}({\cal S}),~~~~~t\in\R
\ee
It is true for $t=0$ by our main assumption (\re{close}) with sufficiently small $d_0>0$.
Then  $S(\si(0))={\bf\Pi} Y(0)$ and $Z(0)=Y(0)-S(\si(0))$ are well defined.
We will prove below that (\re{YtO}) holds if $d_0$ is sufficiently small.
Let us choose an arbitrary $\tilde v$ such that $|v(0)|<\tilde v<1$
and let $\delta=\tilde v-|v(0)|$.
Denote by $r_{-\nu}(\tilde v)$ the positive numbers from Lemma \re{skewpro} iv) which
corresponds to $\al=-\nu$. Then $S(\si)+Z\in {\cal O}_{-\nu}({\cal S})$ if $\si=(b,v)$
with $|v|<\tilde v$ and $\Vert Z\Vert_{-\nu}<r_{-\nu}(\tilde v)$. Note that
$\Vert Z(0)\Vert_{-\nu}<r_{-\nu}(\tilde v)$ if $d_0$ is sufficiently small. Therefore,
$S(\si(t))={\bf\Pi}Y(t)$ and $Z(t)=Y(t)-S(\si(t))$ are well defined for $t\ge 0$ so
small that $|v|<\tilde v$ and $\Vert Z(t)\Vert_{-\nu} < r_{-\nu}(\tilde v)$. This
is formalized by the following standard definition.
\begin{definition}
  $t_*$ is the ``exit time'',
  \be\la{t*}
     t_*=\sup \{t>0: \Vert Z(s)\Vert_{-\nu} < r_{-\nu}(\tilde v),~~
     |v(s)-v(0)|<\delta,~~0\le s\le t \}
  \ee
\end{definition}

One of our main goals is to prove that $t_*=\infty$ if $d_0$ is sufficiently small.
This would follow if we show that
\be\la{Zt}
  \Vert Z(t)\Vert_{-\nu}<r_{-\nu}(\tilde v)/2,~~|v(s)-v(0)|<\delta/2,~~~0\le t < t_*
\ee
Note that
\be\la{Qind}
  |Q(t)|\le\overline Q:= r_{-\nu}(\tilde v), ~~~~~0\le t< t_*
\ee
Now $N(t)$ in (\re{lin}) satisfies, by (\re{N14}), the following estimate,
\be\la{Nest}
  \Vert N(t)\Vert_{\nu}\le C_\nu(\tilde v)\Vert Z(t)\Vert^2_{-\nu},
  \,\,\,0\le t<t_*
\ee
\subsection{Longitudinal dynamics: modulation equations}
From now on we fix the decomposition $Y(t)=S(\si(t))+Z(t)$ for $0<t<t_*$
by setting $S(\si(t))={\bf\Pi} Y(t)$ which is equivalent to the
symplectic orthogonality condition of type (\re{proj}),
\be\la{ortZ}
  Z(t)\nmid{\cal T}_{S(\si(t))}{\cal S},\,\,\,0\le t<t_*
\ee
This allows us to simplify drastically the asymptotic analysis of the dynamical equations
(\re{lin}) for the transversal component $Z(t)$. As the first step, we derive the longitudinal
dynamics, i.e. the ``modulation equations'' for the parameters $\si(t)$.
Let us derive a system of ordinary differential equations for the
vector $\si(t)$. For this purpose, let us write (\re{ortZ}) in the form
\be\la{orth}
  \Om(Z(t),\tau_j(t))=0,\,\,j=1,\dots,6, ~~~~~~~0\le t<t_*
\ee
where the vectors $\tau_j(t)=\tau_j(\si(t))$ span the tangent space
${\cal T}_{S(\si(t))}{\cal S}$. Note that $\si(t)=(b(t),v(t))$, where
\be\la{sit}
  |v(t)|\le \tilde v<1,~~~~~~~~~0\le t<t_*
\ee
by Lemma \re{skewpro} iii). It would be convenient for us to use some other parameters
$(c,v)$ instead of $\si=(b,v)$, where 
\be\la{vw}
c(t)= b(t)-\ds\int^t_0 v(\tau)d\tau,\quad\dot c(t)=\dot b(t)-v(t)=w(t)-v(t), \quad0\le t<t_*
\ee
The following statement can be proved similar to the Lemma 6.2 in \cite{ikv}.
\begin{lemma}\la{mod}
Let $Y(t)$ be a solution to the Cauchy problem (\re{DD}), and (\re{dec}),
(\re{orth}) hold. Then 
\be\la{parameq}
|\dot c(t)|+|\dot v(t)|\le C(\tilde v)\Vert Z\Vert^2_{-\nu}
\ee
\end{lemma}
\subsection{Decay for  transversal dynamics}
In Section 11 we will show that our main Theorem \re{main} can be derived from
the following time decay of the transversal component $Z(t)$:
\begin{prop}\la{pdec}
   Let all conditions of Theorem \re{main} hold. Then $t_*=\infty$, and
   \be\la{Zdec}
     \Vert Z(t)\Vert_{-\nu}\le \ds\fr {C(\rho,\overline v,d_0)}{(1+|t|)^{3/2}},
     ~~~~~t\ge0
   \ee
\end{prop}
We will derive (\re{Zdec}) in Sections 7-10 from our equation (\re{lin}) for the
transversal component $Z(t)$. This equation can be specified using Lemma \re{mod}.
Indeed, the lemma implies that
\be\la{Tta}
  \Vert T(t)\Vert_{\nu}\le C(\tilde v)\Vert Z(t)\Vert^2_{-\nu},
  ~~~~~~~~~0\le t<t_*
\ee
by (\re{TN})  since $w-v=\dot c$. Thus (\re{lin}) becomes the equation
\be\la{reduced}
  \dot Z(t)=A(t)Z(t)+\tilde N(t), ~~~~~~~~~0\le t<t_*
\ee
where $A(t)=A_{v(t),w(t)}$, and $\tilde N(t):=T(t)+N(t)$ satisfies the estimate
\be\la{redN}
   \Vert\tilde  N(t)\Vert_{\nu}\le C(\tilde v,\ov Q)\Vert Z(t)\Vert^2_{-\nu},~~~~
   ~~~~~~~~~0\le t<t_*
\ee
In all remaining part of our paper we will analyze mainly the {\bf basic equation}
(\re{reduced}) to establish the decay (\re{Zdec}). We are going to derive the decay using
the bound (\re{redN}) and the orthogonality condition  (\re{ortZ}).

Similarly \cite{ikv} we reduce the problem to the analysis of the {\it frozen} linear equation,
\be\la{Avv}
  \dot X(t)=A_1X(t), ~~t\in\R
\ee
where $A_1=A_{v_1,v_1}$ with $v_1=v(t_1)$ and a fixed $t_1\in[0,t_*)$. Then  we can apply some
methods of scattering theory and then estimate the error by the method of majorants.

Note, that even for the frozen equation (\re{Avv}), the decay of type  (\re{Zdec})
for all solutions does not hold without  the orthogonality condition of type (\re{ortZ}).
Namely, by  (\re{Atanformv}) the equation (\re{Avv}) admits the {\it secular solutions}
\be\la{secs}
   X(t)=\sum_1^3 C_{j}\tau_j(v)+\sum_1^3 D_j[\tau_j(v)t+\tau_{j+3}(v)]
\ee
which arise  by differentiation of the soliton (\re{sosol}) in the parameters
$a$ and $v$ in the moving coordinate $y=x-v_1t$. Hence, we have to take into account
the orthogonality condition  (\re{ortZ}) in order to avoid the secular solutions.
For this purpose we will apply the corresponding symplectic orthogonal projection
which kills the ``runaway solutions''  (\re{secs}).
\begin{remark}
{\rm
  The solution (\re{secs}) lies in the tangent space  ${\cal T}_{S(\si_1)}{\cal S}$
  with $\si_1=(b_1,v_1)$ (for an arbitrary $b_1\in\R$) that suggests an unstable
  character of the nonlinear dynamics {\it along the solitary manifold}.
}
\end{remark}
\begin{definition}
  i) For $v\in V$, denote by ${\bf\Pi}_v$ the symplectic orthogonal projection
  of ${\cal E}$ onto the tangent space ${\cal T}_{S(\si)}{\cal S}$, and
  ${\bf P}_v={\bf I}-{\bf\Pi}_v$.\\
  ii) Denote by ${\cal Z}_v={\bf P}_v{\cal E}$ the space symplectic orthogonal to
  ${\cal T}_{S(\si)}{\cal S}$ with $\si=(b,v)$.
\end{definition}
Note that by the linearity,
\be\la{Piv}
  {\bf\Pi}_vZ=\sum{\bf\Pi}_{jl}(v)
  \tau_j(v)\Om(\tau_l(v),Z),~~~~~~~~~~Z\in{\cal E}
\ee
with some smooth coefficients ${\bf\Pi}_{jl}(v)$. Hence, the projector ${\bf\Pi}_v$,
in the variable $y=x-b$, does not depend on $b$, and this explains the choice
of the subindex in ${\bf\Pi}_v$ and ${\bf P}_v$.

Now we have the symplectic orthogonal decomposition
\be\la{sod}
  {\cal E}={\cal T}_{S(\si)}{\cal S}+{\cal Z}_v,~~~~~~~\si=(b,v),
\ee
and the symplectic orthogonality  (\re{ortZ}) can be written in the equivalent forms,
\be\la{PZ}
  {\bf\Pi}_{v(t)} Z(t)=0,~~~~{\bf P}_{v(t)}Z(t)= Z(t),~~~~~~~~~0\le t<t_*
\ee
\begin{remark}\la{rZ}
{\rm
  The tangent space ${\cal T}_{S(\si)}{\cal S}$ is invariant under the operator
  $A_{v,v}$ by Lemma \re{ceig} i), hence the space  ${\cal Z}_v$ is also invariant
  by (\re{com}): $A_{v,v}Z\in {\cal Z}_v$ for {\it sufficiently smooth}  $Z\in{\cal Z}_v$.
}
\end{remark}
Below in section 12-18 we will prove the following proposition which will be one of the main
ingredients for proving \eqref{Zdec}. Let us consider the Cauchy problem for the equation \eqref{Avv}
with $A=A_{v,v}$ for a fixed $v\in V$. Recall that the parameter $\nu>5/2$ is also fixed.
\begin{prop}\la{lindecay}
  Let the conditions \eqref{ro}- \eqref{W} hold, $|v|\le\tilde v<1$, and $X_0\in{\cal E}$.
  Then\\
  i) Equation (\re{Avv}), with $A=A_{v,v}$, admits the unique solution
  $e^{At}X_0:=X(t)\in C(\R, {\cal E})$ with the initial condition $X(0)=X_0$.\\
  ii) For $X_0\in{\cal Z}_v\cap{\cal E_{\nu}}$, the decay holds,
  \be\la{frozenest}
     \Vert e^{At}X_0\Vert_{-\nu}\le
     \fr{C_{\nu}(\rho,\tilde v)}{(1+|t|)^{3/2}}\Vert X_0\Vert_{\nu},\quad t\in\R
  \ee
\end{prop}
\setcounter{equation}{0}
\section{Frozen transversal dynamics}

Now  let us fix an arbitrary $t_1\in [0,t_*)$, and
rewrite the equation (\re{reduced}) in a ``frozen form''
\be\la{froz}
   \dot Z(t)=A_1Z(t)+(A(t)-A_1)Z(t)+\tilde N(t),\,\,\,~~~~0\le t<t_*
\ee
where $A_1=A_{v(t_1),v(t_1)}$ and
\begin{equation*}
   A(t)-A_1\! =\!\left(\!\!
  \begin{array}{cccc}
   [w(t)\!-\!v(t_1)]\cdot \na  &                    0    & 0 & 0 \\
   0                       &  [w(t)\!-\!v(t_1)]\cdot \na  &0 & 0 \\
   0                      & 0 & 0 & B_{v(t)}\!-\!B_{v_1(t)}\\
   0  & 0& \langle\na(\psi_{v(t)j}\!-\!\psi_{v(t_1)j}),\na\rho_j\rangle & 0
  \end{array}\!\!\!\right)
\end{equation*}
The next trick is important since it allows us to kill the ``bad terms''
 $[w(t)\!-\!v(t_1)]\cdot \na$ in the operator $A(t)-A_1$.
\begin{definition}\la{d71}
  Let us change the  variables $(y,t)\mapsto (y_1,t)=(y+d_1(t),t)$, where
  \be\la{dd1}
     d_1(t):=\int_{t_1}^t(w(s)-v(t_1))ds, ~~~~0\le t\le t_1
  \ee
\end{definition}
Next define
\be\la{Z1}
   Z_1(t):=(\Psi_1(y_1-d_1(t),t),\Psi_2(y_1-d_1(t),t),Q(t),P(t))
\ee
Then we obtain the final form of the ``frozen equation'' for the transversal dynamics
\be\la{redy1}
  \dot Z_1(t)=A_1Z_1(t)+B_1(t)Z_1(t)+\tilde N_1(t),\,\,\,0\le t\le t_1
\ee
where $\tilde N_1(t)=\tilde N(t)$ expressed in terms of $y=y_1-d_1(t)$,  and
$$
  B_1(t)=\left(
  \begin{array}{cccc}
  0  & 0  & 0 & 0 \\
  0  & 0  & 0 & 0 \\
  0  & 0  & 0 &  B_{v(t)}-B_{v_1(t)}\ \\
  0  & 0  & \langle\na(\psi_{v(t)j}\!-\!\psi_{v(t_1)j}),\na\rho_j\rangle & 0
  \end{array}
  \right)
$$
Let us estimate the ``remaining terms'' $B_1(t)Z_1(t)$ and $\tilde N_1(t)$.
\begin{lemma}\la{Best} The bound holds
\be\la{Bestc}
\Vert B_1(t)Z_1(t)\Vert_{\nu}\le C(\tilde v)\Vert Z(t)\Vert_{-\nu}\int_t^{t_1}
\Vert Z(s)\Vert^2_{-\nu}ds,\quad 0\le t\le t_1
\ee
\end{lemma}
\begin{proof}
Lemma \ref{mod} implies
$$
|B_{v(t)}-B_{v_1(t)}|\le|\int\limits_{t_1}^{t}\dot v(s)\cdot\na_{v}B_{v(s)}ds|
\le C(\tilde v)\int\limits_{t_1}^{t}\Vert Z(s)\Vert^2_{-\nu}ds
$$
similarly,
$$
|\langle\na(\psi_{v(t)j}-\psi_{v(t_1)j}),\na\rho_j|
\le C(\tilde v)\int\limits_{t_1}^{t}\Vert Z(s)\Vert^2_{-\nu}ds
$$
Therefore,
\beqn\nonumber
\Vert B_1(t)Z_1(t)\Vert_{\nu}&=&
|\langle\na(\psi_{v(t)j}-\psi_{v(t_1)j}),\na\rho_j\rangle Q_1(t)|+
|(B_{v(t)}-B_{v_1(t)})P_1(t)|\\
\nonumber
&\le& C(\tilde v)(|Q(t)|+|P(t)|)\int\limits_{t_1}^{t}\Vert
Z(s)\Vert^2_{-\nu}ds\le C(\tilde v)\Vert Z(t)\Vert_{-\nu}\int_t^{t_1}
\Vert Z(s)\Vert^2_{-\nu}ds
\eeqn
\end{proof}
\begin{lemma}\la{destc}
The   bounds hold 
\be\la{dest1}
\Vert\tilde N_1(t)\Vert_{\nu}\le C(\tilde v,\ov
Q)(1+|d_1(t)|)^{\nu}\Vert Z(t)\Vert^2_{-\nu},\quad 0\le t\le t_1
\ee
\end{lemma}
\begin{proof}
For any
$\Phi\in L^2_{\alpha}$ and $d\in\R^3$ we have
\beqn\nonumber
\Vert\Phi(y-d)\Vert_{\alpha}^2&=&\int|\Phi(y-d)|^2(1+|y|)^{2\alpha}dy
=\int|\Phi(y)|^2(1+|y+d|)^{2\alpha}dy\\
\nonumber
&\le&\int|\Phi(y)|^2(1+|y|)^{2\alpha}(1+|d|)^{2\alpha}dy
\le(1+|d|)^{2\alpha} \Vert\Phi\Vert_{\alpha}^2,\quad\alpha\in\R
\eeqn
Hence,  the bound  (\ref{dest1}) follows.
\end{proof}
\setcounter{equation}{0}
\section{Integral inequality}
The equation (\re{redy1}) can be written in the integral form:
\be\la{Z1duh}
  Z_1(t)=e^{A_1t}Z_1(0)+\int_0^te^{A_1(t-s)}[B_1Z_1(s)+\tilde N_1(s)]ds,\quad 0\le t\le t_1
\ee
Now we apply the symplectic orthogonal projection ${\bf P}_1:={\bf P}_{v(t_1)}$ to both sides
of \eqref{Z1duh}:
$$
{\bf P}_1 Z_1(t)
=e^{A_1t}{\bf P}_1Z_1(0)+\int_0^te^{A_1(t-s)}{\bf P}_1[B_1Z_1(s)+\tilde N_1(s)]ds
$$ 
The projector  ${\bf P}_1$
commutes with the group $e^{A_1t}$
since the space ${\cal Z}_1:={\bf P}_1{\cal E}$ is invariant with respect to
$e^{A_1t}$ by  Remark \re{rZ}. Applying (\re{frozenest}) we obtain that
$$
  \Vert {\bf P}_1Z_1(t)\Vert_{-\nu}
  \le C\fr{\Vert {\bf P}_1Z_1(0)\Vert_{\nu}}{(1+t)^{3/2}}
  +C\int_0^t\fr{\Vert{\bf P}_1[B_1Z_1(s)+\tilde N_1(s)]\Vert_{\nu}~ds}{(1+|t-s|)^{3/2}}.
$$
The operator ${\bf P}_1={\bf I}-{\bf\Pi}_1$ is continuous in ${\cal E}_\nu$ by
(\re{Piv}). Hence,  (\re{Bestc})-(\re{dest1}) imply
\begin{multline}\la{duhest}
\Vert {\bf P}_1Z_1(t)\Vert_{-\nu}\le\fr{C(\overline d_1(0))}{(1+t)^{3/2}}\Vert Z(0)\Vert_{\nu}\\
   +C(\overline d_1(t))\int_0^t\!\!\fr1{(1+|t-s|)^{3/2}}\left[\Vert Z(s)\Vert_{-\nu}
   \int_s^{t_1}\Vert Z(\tau)\Vert^2_{-\nu}d\tau+\Vert Z(s)\Vert^2_{-\nu}\right]ds,
   \quad 0\le t\le t_1.
 \end{multline}
where $\overline d_1(t):=\sup_{0\le s\le t} |d_1(s)| $.
Let us introduce the ``majorant''
\be\la{maj}
  m(t):=
  \sup_{s\in[0,t]}(1+s)^{3/2}\Vert Z(s)\Vert_{-\nu},\quad t\in [0,t_*).
\ee
Now we reduce further the exit time. 
Denote by $\ve<1$ a fixed positive number which we
will specify below.
\begin{definition} $t_{*}'$ is the exit time
  \be\la{t*'}
  t_*'=\sup \{t\in[0,t_*):m(s)\le \ve,~~0\le s\le t\}.
\ee
\end{definition}
To estimate $d_1(t)$,  note that
 \be\la{wen}
   w(s)-v(t_1)=w(s)-v(s)+v(s)-v(t_1)= \dot c(s)+\int_s^{t_1}\dot v(\tau)d\tau
 \ee
 by (\re{vw}). Hence, (\ref {dd1}), Lemma \ref{mod} and the definition
 (\re{maj}) imply that for $t_1<t_*'$
\begin{multline}\label{d1est}
   |d_1(t)| =|\int_{t_{1}}^t(w(s)-v(t_1))ds|\le
   \int_t^{t_{1}}\left( |\dot c(s)|+\int_s^{t_1}|\dot v(\tau)|d\tau\right)ds\\\\
   \le  C(\tilde v)m^2(t_1)\int_t^{t_1}\left(\frac 1{(1+s)^3}
   +\int_s^{t_1}\frac{d\tau}{(1+\tau)^3}\right) ds
  \le C(\tilde v)m^2(t_1)\le C(\tilde v),\quad 0\le t\le t_1
 \end{multline}
Now we can to replace $C(\overline d_1)$ with $C(\tilde v)$ in (\ref{duhest}): 
for $t_1<t_*'$
\begin{multline}\la{duhestri}
  \! \Vert {\bf P}_1Z_1(t)\Vert_{-\nu}\le\fr{C(\tilde v)}{(1+t)^{3/2}}\Vert Z(0)\Vert_{\nu}\\
  \!\! +C(\tilde v)\int_0^t\!\fr1{(1+|t-s|)^{3/2}}\left[\Vert Z(s)\Vert_{-\nu}
   \int_s^{t_1}\Vert Z(\tau)\Vert^2_{-\nu}d\tau+
   \Vert Z(s)\Vert^2_{-\nu}\right]ds,\quad 0\le t\le t_1
\end{multline}
\setcounter{equation}{0}
\section{Symplectic orthogonality}
Finally, we are going to change ${\bf P}_1Z_1(t)$ by $Z(t)$ in the left hand side of
(\re{duhestri}). We will prove that it is possible using again that $d_0\ll 1$ in (\re{close}).
\begin{lemma}\la{Z1P1Z1}(cf.\cite{ikv})
  For sufficiently small $\ve>0$, we have for  $t_1<t_*'$
  \be\la{Z1P1est}
     \Vert Z(t)\Vert_{-\nu}\le C\Vert {\bf P}_1Z_1(t)\Vert_{-\nu},
     ~~~~~~~~0\le t \le t_1,
  \ee
  where $C$ depends only on $\rho$ and $\overline v$.
\end{lemma}
\begin{proof}
 Since $|d_1(t)|\le C$ for $t\le t_1<t_*'$ then 
$ \Vert Z(t)\Vert_{-\nu}\le C \Vert Z_1(t)\Vert_{-\nu}$, and  it suffices
 to prove that
 \be\la{Z1P1ests}
   \Vert Z_1(t)\Vert_{-\nu}\le 2\Vert{\bf P}_1Z_1(t)\Vert_{-\nu},\quad 0\le t\le t_1.
 \ee
 Recall that ${\bf P}_1Z_1(t)=Z_1(t)-{\bf\Pi}_{v(t_1)}Z_1(t)$.
 Then estimate \eqref{Z1P1ests} will follow from
 \be\la{Z1P1estf}
 \Vert{\bf\Pi}_{v(t_1)}Z_1(t)\Vert_{-\nu}
 \le\fr 12\Vert Z_1(t)\Vert_{-\nu},\quad 0\le t\le t_1.
 \ee
 Symplectic orthogonality \eqref{PZ} implies
 \be\la{PZ0r}
   {\bf\Pi}_{v(t),1}Z_1(t)=0,~~~~t\in[0,t_1],
 \ee
 where ${\bf\Pi}_{v(t),1}Z_1(t)$ is ${\bf\Pi}_{v(t)}Z(t)$ expressed in terms of
 the variable $y_1=y+d_1(t)$. Hence, (\re{Z1P1estf}) follows from (\re{PZ0r}) if the difference
 ${\bf\Pi}_{v(t_1)}-{\bf\Pi}_{v(t),1}$ is small uniformly in $t$, i.e.
 \be\la{difs}
   \Vert{\bf\Pi}_{v(t_1)}-{\bf\Pi}_{v(t),1}\Vert<1/2,~~~~~~~0\le t\le t_1.
 \ee
 It remains to justify (\ref{difs}) for small enough $\varepsilon>0.$
 Formula (\ref{Piv}) implies 
 \be\la{Piv1}
  {\bf\Pi}_{v(t),1}Z_{1}(t)=\sum{\bf\Pi}_{jl}(v(t))\tau_{j,1}(v(t))
  \Omega (\tau_{l,1}(v(t)),Z_1(t)),
 \ee
 where $\tau_{j,1}(v(t))$ are the vectors $\tau_j(v(t))$ expressed in the variables $y_{1}$.
 Since $|d_1(t)|\le C$ and $\nabla\tau_j$
 are smooth and fast decaying at infinity functions, then
 \be\la{011}
   \Vert \tau_{j,1}(v(t))-\tau_j(v(t))\Vert_\nu\le C|d_1(t)|^\nu\le C,\quad 0\le t\le t_1
 \ee
 for all $j=1,2,\dots,6$. Furthermore,
 \[   \tau_j(v(t))-\tau_j(v(t_1)) =\int_t^{t_1}\dot v(s)\cdot\nabla_v\tau_j(v(s))ds    \]
 and therefore
 \be\la{012}
   \Vert \tau_j(v(t))-\tau_j(v(t_1))\Vert_\nu
   \le C\int_t^{t_1}|\dot v(s)|ds,\quad 0\le t\le t_1
 \ee
 Similarly,
 \be\la{013}
   |{\bf\Pi}_{jl}(v(t))-{\bf\Pi}_{jl}(v(t_1))|= |\int_t^{t_1}\dot v(s)\cdot\nabla_v
   {\bf\Pi}_{jl}(v(s))ds|\le C\int_t^{t_1}|\dot v(s)|ds,~~~~0\leq t\le t_1
 \ee
 since $|\nabla_v{\bf\Pi}_{jl}(v(s))|$ is uniformly bounded by (\re{sit}). Hence, the bounds
 (\ref{difs}) will follow from (\ref{Piv}), (\ref{Piv1}) and (\ref{011})-(\ref{013}) if we
 establish that  the integral in the right hand side of (\ref{012}) can be made
 as small as we please by choosing $\ve >0$ small enough. Indeed,
\be\la{tvjest}
   \int_t^{t_1}|\dot v(s)|ds\le Cm^2(t_1)\int_t^{t_{1}}\frac{ds}{(1+s)^3}\le C\ve^2,\quad
   0\le t\le t_1
\ee
\end{proof}
\setcounter{equation}{0}
\section{Decay of transversal component}
Here we prove Proposition \re{pdec}.
\\
{\it Step i)} We fix $0<\ve<1$ and $t'_*=t'_*(\ve)$ for which Lemma \re{Z1P1Z1} holds.
Then the bound of type (\re{duhestri}) holds with
$\Vert {\bf P}_1Z_1(t)\Vert_{-\nu}$ in the left hand side replaced by
$\Vert Z(t)\Vert_{-\nu}$~:
\begin{multline}\la{duhestrih}
   \Vert Z(t)\Vert_{-\nu}\le\fr{C}
   {(1+t)^{3/2}}\Vert Z(0)\Vert_{\nu}\\
   +C\int_0^t\fr1{(1+|t-s|)^{3/2}}\left[\Vert Z(s)\Vert_{-\nu}
   \int_s^{t_1}\Vert Z(\tau)\Vert^2_{-\nu}d\tau+
   \Vert Z(s)\Vert^2_{-\nu}\right]ds,\quad
   0\le t\le t_1
\end{multline}
for $t_1<t_*'$. This implies an integral inequality for the majorant $m(t)$
defined in \eqref{maj}.
Namely, multiplying both sides of (\re{duhestrih}) by $(1+t)^{3/2}$,
and taking the supremum in $t\in[0,t_1]$, we get
\begin{multline*}
\!\!\!\!  m(t_1)\le C\Vert Z(0)\Vert_{\nu}+C\sup_{t\in[0,t_1]}\ds
  \int_0^t\fr{(1+t)^{3/2}}{(1+|t-s|)^{3/2}}\left[\fr{m(s)}{(1+s)^{3/2}}
  \int_s^{t_1}\fr{m^2(\tau)d\tau}{(1+\tau)^{3}}+\fr{m^2(s)}{(1+s)^{3}}\right]ds
\end{multline*}
for $t_1\le t_*'$. Taking into account that $m(t)$ is a monotone
increasing function, we get
\be\la{mest}
   m(t_1)\le C\Vert Z(0)\Vert_{\nu}+C[m^3(t_1)+m^2(t_1)]I(t_1),\quad t_1\le t_*'
   \ee
where
\begin{multline*}
  I(t_1)=\sup_{t\in[0,t_1]}
  \int_0^{t}\fr{(1+t)^{3/2}}{(1+|t-s|)^{3/2}}\left[\fr1{(1+s)^{3/2}}
  \int_s^{t_1}\fr{d\tau}{(1+\tau)^{3}}+\fr1{(1+s)^3}\right]ds\le \overline I<\infty
\end{multline*}
Therefore, (\re{mest}) becomes
\be\la{m1est}
  m(t_1)\le C\Vert Z(0)\Vert_{\nu}+C\overline I[m^3(t_1)+m^2(t_1)],~~~~ t_1<t_*'
\ee
This inequality implies that $m(t_1)$ is bounded for $t_1<t_*'$, and moreover,
\be\la{m2est}
  m(t_1)\le C_1\Vert Z(0)\Vert_{\nu},~~~~~~~~~t_1<t_*'
\ee
since $m(0)=\Vert Z(0)\Vert_{\nu}$ is sufficiently small by (\re{closeZ}).
\\
{\it Step ii)} The constant $C_1$ in the estimate (\re{m2est}) does not depend on
$t_*$ and $t_*'$ by Lemma \re{Z1P1Z1}. We choose $d_0$ in (\re{close})
so small that $\Vert Z(0)\Vert_{\nu}<\ve/(2C_1)$. It is possible due to (\re{closeZ}).
Then the estimate (\re{m2est}) implies that $t'_*=t_*$ and therefore
(\re{m2est}) holds for all $t_1<t_*$.
Further,
$$
|v(t)-v(0)|\le\int\limits_0^t|\dot v(s)|ds\le Cm^2(t)\int\limits_0^t
\frac{ds}{(1+s)^3}\le Cm^2(t)
$$
Hence the both inequalities  (\re{Zt}) also holds if $\Vert Z(0)\Vert_{\nu}$ is
sufficiently small by \eqref{maj}.
Finally, this implies that $t_*=\infty$, hence also $t'_*=\infty$ and
(\re{m2est}) holds for all $t_1>0$ if $d_0$ is small enough.
It complete the proof of Proposition \re{pdec}. 
\section{Soliton asymptotics}
\setcounter{equation}{0}

Here we prove our main Theorem \re{main} under the assumption that the decay (\re{Zdec}) holds.
First we will prove the asymptotics (\re{nas}) for the vector components, and afterwards
the asymptotics (\re{nasf}) for the fields.
\\
{\bf Asymptotics for the vector components}.
From (\re{addeq}) we have $\dot q=\dot b+\dot Q$, and from (\re{reduced}), (\re{redN}), (\re{AA})
it follows that $\dot Q=P+{\cal O}(\Vert Z\Vert^2_{-\nu})$. Thus,
\be\la{dq}
   \dot q=\dot b+\dot Q=v(t)+\dot c(t)+P(t)+{\cal O}(\Vert Z\Vert^2_{-\nu})
\ee
Bounds (\re{parameq}) and  (\re{Zdec}) imply that
\be\la{bv}
  |\dot c(t)|+|\dot v(t)|\le \ds\fr {C_1(\rho,\overline v,d_0)}{(1+t)^{3}},
  ~~~~~~t\ge0
\ee
Therefore, $c(t)=c_+ +{\cal O}(t^{-2})$ and $v(t)=v_+ +{\cal O}(t^{-2})$, $t\to\infty$. Since
$|P|\le\Vert Z\Vert_{-\nu}$, the estimate (\re{Zdec}), and (\re{dq})-(\re{bv}),  imply that
\be\la{qbQ}
  \dot q(t)=v_++{\cal O}(t^{-3/2})
\ee
Similarly,
\be\la{bt}
  b(t)=c(t)+\ds\int_0^tv(s)ds=v_+t+a_++{\cal O}(t^{-1})
\ee
hence the second part of (\re{qq}) follows:
\be\la{qbQ2}
  q(t)=b(t)+Q(t)=v_+t+a_++{\cal O}(t^{-1})
\ee
since $Q(t)={\cal O}(t^{-3/2})$ by  (\re{Zdec}).
\\
{\bf Asymptotics for the fields}.
For the field part of
the solution $\psi(x,t)$ let us define the accompanying soliton
field as $\psi_{\rm v(t)}(x-q(t))$, where we define now ${\rm v}(t)=\dot q(t)$, cf. (\re{dq}).
Then for the difference $z(x,t)=\psi(x,t)-\psi_{\rm v(t)}(x-q(t))$ we obtain the equation
$$
\dot z(x,t)=[-\al_j\pa_j-i\beta m]z(x,t)-i\dot{\rm v}\cdot\na_{\rm v}\psi_{{\rm v}(t)}(x-q(t))
$$
Then
\be\la{eqacc}
z(t)=W_0(t)z(0)-\int_0^tW_0(t-s)[i\dot{\rm v}(s)\cdot\na_{\rm v}\psi_{{\rm v}(s)}(\cdot-q(s))]ds
\ee
To obtain the asymptotics (\re{nasf}) it suffices to prove that 
$z(t)=W_0(t)\phi_++r_+(t)$ with
some $\phi_+\in L^2_0$ and $\Vert r_+(t)\Vert_{0}={\cal O}(t^{-1/2})$.
This is equivalent to
\be\la{Sme}
  W_0(-t)z(t)=\phi_++r_+'(t)
\ee
where $\Vert r_+'(t)\Vert_{0}={\cal O}(t^{-1/2})$ since $W_0(t)$ is a unitary group in 
$L^2_0$ by the charge conservation for the free Dirac equation.
Finally, (\re{Sme}) holds since (\re{eqacc}) implies 
$$
  W_0(-t)z(t)=z(0)-\int_0^t W_0(-s)f(s)ds,\quad
  f(s)=i\dot{\rm v}(s)\cdot\na_{\rm v}\psi_{{\rm v}(s)}(\cdot-q(s))
$$
where the integral  in the right hand side   converges in $L^2_0$
with the rate ${\cal O}(t^{-1/2})$. The latter holds since
$\Vert  W_0(-s)f(s)\Vert_{0} ={\cal O}(s^{-3/2})$ by the unitarity of $W_0(-s)$ and
the decay rate $\Vert f(s)\Vert_{0} ={\cal O}(s^{-3/2})$. Let us prove this rate of decay.
It suffices to prove that $|\dot {\rm v}(s)|={\cal O}(s^{-3/2})$,
or equivalently $|\dot p(s)|={\cal O}(s^{-3/2})$. Substitute (\re{add}) to the last
equation of (\re{D}) and obtain
\beqn\nonumber
\dot p(t)&=&\Re\int\left[\psi_{v(t)}(x-b(t))+\Psi(x-b(t),t)\right]\na\rho(x-b(t)-Q(t))dx\\
\nonumber
&=&\Re\int\psi_{v(t)}(y)\na\rho(y)dy
+\Re\int\psi_{v(t)}(y)\left[\na\rho(y-Q(t))-\na\rho(y)\right]dy\\
\nonumber
&+&\Re\int\Psi(y,t)\na\rho(y-Q(t))dy
\eeqn
The first integral in the right hand side is zero by the stationary equations (\re{stfch}).
The second integral is ${\cal O}(t^{-3/2})$, since $Q(t)={\cal O}(t^{-3/2})$, and by the conditions
(\re{ro}) on $\rho$. Finally, the third integral is ${\cal O}(t^{-3/2})$
by the estimate (\re{Zdec}). The proof is complete.
\setcounter{equation}{0}
\section{Decay for  linearized dynamics}
In remaining sections we prove Proposition \re{lindecay}.
Here we discuss our general strategy of the proof. We apply the Fourier-Laplace transform
\be\la{FL}
  \tilde X(\lambda)=\int_0^\infty e^{-\lambda t}X(t)dt,~~~~~~~\Re\lambda>0
\ee
to (\re{Avv}).  According to Proposition \re{lindecay}, we expect that
the solution $X(t)$ is bounded in the norm $\Vert\cdot\Vert_{-\nu}$.
Then the integral (\re{FL}) converges and is analytic for $\Re\lambda>0$.
We will write $A$ and $v$ instead of $A_1$ and $v_1$ in all remaining part of the paper.
After the Fourier-Laplace transform  (\re{Avv}) becomes
\be\la{FLA}
  \lambda\tilde X(\lambda)=A\tilde X(\lambda)+X_0, \quad\Re\lambda>0
\ee
Let us stress that (\re{FLA}) is equivalent to the Cauchy problem for the functions
$X(t)\in C_b([0,\infty);{\cal E}_{-\nu})$. Hence the solution $X(t)$ is given by
\be\la{FLAs}
  \tilde X(\lambda)=-(A-\lambda)^{-1}X_0,~~~~~~~~\Re\lambda>0
\ee
if the resolvent $R(\lambda)=(A-\lambda)^{-1}$ exists for $\Re\lambda>0$.

Let us comment on our following strategy in proving  the decay (\re{Zdec}). First, we will
construct the resolvent  $R(\lambda)$ for $\Re\lambda>0$ and prove that it is a continuous operator
in ${\cal E}_{-\nu}$. Then $\tilde X(\lambda)\in{\cal E}_{-\nu}$ and is an analytic function for
$\Re\lambda>0$. Second, we have to justify that there exist a (unique) function
$X(t)\in C([0,\infty);{\cal E}_{-\nu})$ satisfying (\re{FL}).

The analyticity of $\tilde X(\lambda)$ and Paley-Wiener arguments (see \ci{EKS}) should provide
the existence of a ${\cal E}_{-\nu}$ - valued distribution $X(t)$, $t\in\R$,
with a support in $[0,\infty)$. Formally,
\be\la{FLr}
\Lambda^{-1}\tilde X=X(t)=\fr 1{2\pi}\int_\R e^{i\om t}\tilde X(i\om+0)d\om,\quad t\in\R
\ee
However, to check the continuity of $X(t)$ for $t\ge 0$, we need additionally a bound for
$\tilde X(i\om+0)$ at large $|\om|$. Finally, for the time decay of $X(t)$, we need an additional
information on the smoothness and decay of $\tilde X(i\om+0)$. More precisely,
we should prove that the function $\tilde X(i\om+0)$ \\
i) is smooth outside $\om=0$ and $\om=\pm\mu$, where $\mu=\mu(v)>0$,
\\
ii) decays in a certain sense as $|\om|\to\infty$.
\\
iii) admits the Puiseux expansion at $\om=\pm\mu$.
\\
iv) is analytic at $\om=0$ if $X_0\in{\cal Z}_v:={\bf P}_v{\cal E}$ and $X_0\in{\cal E}_\nu$.
\\
Then the decay (\re{Zdec}) would follow from the Fourier-Laplace representation (\re{FLr}).

\setcounter{equation}{0}
\section{Solving the linearized equation}
Here we construct the resolvent as a bounded operator in ${\cal E}_{-\nu}$ for $\Re\lambda>0$.
We will write $(\tilde\Psi_1,\tilde\Psi_2, \tilde Q, \tilde P)$ instead of
$(\tilde\Psi_1(y,\lambda),\tilde\Psi_2(y,\lambda),\tilde Q(\lambda),\tilde P(\lambda))$
to simplify the notations. Then (\re{FLA}) reads
$$
(A-\lambda)\left(
\begin{array}{c}
\tilde \Psi_1 \\\tilde \Psi_2 \\ \tilde Q \\ \tilde P
\end{array}
\right)=-\left(
\begin{array}{c}
\Psi_{01} \\ \Psi_{02}  \\ Q_0 \\ P_0
\end{array}
\right)
$$
It is  the system of equations
\be\la{eq1}
\left.\begin{array}{r}
(-\al_1\pa_1-\al_3\pa_3+v\cdot\na-\lambda)\tilde\Psi_1+(\beta m+\tilde\al_2\pa_2)\tilde\Psi_2
-\tilde Q\cdot\na\rho_2=-\Psi_{01}
\\\\
-(\beta m+\tilde\al_2\pa_2)\tilde\Psi_1+(-\al_1\pa_1-\al_3\pa_3+v\cdot\na-\lambda)\tilde\Psi_2
+\tilde Q\cdot\na\rho_1=-\Psi_{02}
\\\\
B_v\tilde P-\lambda \tilde Q=-Q_0
\\\\
-\langle\na\tilde\Psi_j,\rho_j\rangle+\langle\na\psi_{vj},\tilde Q\cdot\na\rho_j
\rangle-\lambda\tilde P=-P_0
\end{array}\!\right|
\ee
{\it Step i)} Let us study the first two equations.  First, we compute the matrix integral kernel 
$G_\lam(y-y')$ of the Green operator
\be\la{F1}
G_\lam=\left(\begin{array}{cc}
 -\al_1\pa_1-\al_3\pa_3+v\cdot\na-\lambda  & \beta m+\tilde\al_2\pa_2 \\
 - \beta m-\tilde\al_2\pa_2  &-\al_1\pa_1-\al_3\pa_3+v\cdot\na-\lambda 
  \end{array}\right)^{-1}
\ee
 In Fourier space
$$
 \hat G_{\lambda}(k)=\left(\begin{array}{cc}
  i\al_1k_1+i\al_3k_3-ivk-\lambda & \beta m-\al_2k_2 \\
  -\beta m+\al_2k_2 & i\al_1k_1+i\al_3k_3-ivk-\lambda
  \end{array}\right)^{-1}, \quad vk=\sum\limits_{j=1}^3v_jk_j
$$
To invert the matrix, let us solve the system
\be\la{f1}
 \left.\begin{array}{r}
af_1 +bf_2=g_1\\
-bf_1+af_2=g_2
 \end{array}\right|
\ee
where $a=i\al_1k_1+i\al_3k_3-ivk-\lambda$, $b=\beta m-\al_2k_2$.
Multiplying the first equation of (\re{f1}) by $c=-i\al_1k_1-i\al_3k_3-ivk-\lambda$ 
and the second equation by $-b$, we obtain
\be\la{f2}
 \left.\begin{array}{rcr}
 caf_1+cbf_2&=&cg_1\\
 b^2f_1-cbf_2&=&-bg_2
\end{array}\right|
\ee
since $ba=cb$ by the anticommutations (\re{aij}). Further, $b^2+ac=k^2+m^2+(ivk+\lambda)^2$.
Therefore, summing up the equations (\re{f2}), we obtain that 
$$
f_1=\frac{cg_1-bg_2}{k^2+m^2+(ivk+\lambda)^2}
$$
Similarly, we obtain
$$
f_2=\frac{bg_1+cg_2}{k^2+m^2+(ivk+\lambda)^2}
$$
Hence
\be\la{Delam}
  \hat G_{\lambda}(k)\!=\!\frac1{k^2\!+\!m^2\!+\!(ivk\!+\!\lambda)^2}
    \left(\!\!\begin{array}{cc}
  - i\al_1k_1\!-\!i\al_3k_3\!-\!ivk\!-\!\lambda & -\beta m+\al_2k_2 \\
  \beta m-\al_2k_2         & - i\al_1k_1\!-\!i\al_3k_3\!-\!ivk\!-\!\lambda
  \end{array}\!\!\right)
\ee
Taking the inverse Fourier transform we obtain 
\be\la{Gelam}
  G_{\lambda}(y)=\left(\begin{array}{cc}
 (\al_1\pa_1+\al_3\pa_3+v\cdot\na-\lambda) & -(\beta m+\tilde\al_2\pa_2)\\
 (\beta m+\tilde\al_2\pa_2) & (\al_1\pa_1+\al_3\pa_3+v\cdot\na-\lambda)
  \end{array}\right)g_{\lambda}(y)
\ee
where
\be\la{dete}
g_\lambda(y)=F^{-1}_{k\to y}\ds\fr{1}{k^2+m^2+(ivk+\lambda)^2},~~~y\in\R^3
\ee
Note that denominator in RHS \eqref{dete} does not vanish for $\Re\lambda>0$
since $|v|<1$. This implies
\begin{lemma}\la{cres}
  The operator $G_{\lambda}$ with the integral kernel $G_{\lambda}(y-y')$,
  is continuous operator
  $L^2_0\oplus L^2_0\to L^2_0\oplus L^2_0$ for $\Re\lambda>0$.
\end{lemma}
From now on we use the system of coordinates in $y$-space in which
$v=(|v|,0,0)$, hence $vk=|v|k_1$. Let us compute the function $g_{\lambda}(y)$.
One has
$$
k^2+m^2+(i|v|k_1+\lam)^2=\frac 1{\gamma^2}k_1^2+k_2^2+k_3^2+
2i|v|k_1\lam+\lam^2+m^2=\frac 1{\gamma^2}(k_1+i\gamma^2|v|\lam)^2+k_2^2+k_3^2+ \varkappa^2
$$
where
 \be\la{kappa}
\gamma=1/\sqrt{1-v^2},\quad\varkappa^2=\fr{v^2\lam^2}{1-v^2}+\lam^2+m^2=\fr{\lam^2}{1-v^2}+m^2
=\gamma^2(\lam^2+\mu^2),\quad\mu:=m/\gamma
\ee
Hence formally,
\beqn\nonumber
g_{\lam}(y)&=&\frac{1}{(2\pi)^{3/2}}\int\!\frac{e^{-iky}dk}
{\frac 1{\gamma^2}(k_1+i\gamma^2|v|\lam)^2+k_2^2+k_3^2+\kappa^2}
=\frac{e^{-\gamma^2|v|\lam y_1}}{(2\pi)^{3/2}}\int\!\frac{e^{-iky}dk}
{\frac 1{\gamma^2}k_1^2+k_2^2+k_3^2+\kappa^2}\\
\la{gl}
&=&\frac{\gamma e^{-\gamma|v|\lam\tilde y_1}}{(2\pi)^{3/2}}\int\frac{e^{-ik\tilde y}dk}
{k_1^2+k_2^2+k_3^2+\kappa^2}=\gamma e^{-\gamma|v|\lam\tilde y_1}R(\tilde y,-\kappa^2)
\eeqn
Here $\tilde y_1=\gamma y_1$, $\tilde y=(\gamma y_1,y_2,y_3)$, and
$R(y-y',\zeta)$ is the integral kernel of the  operator
$R(\zeta)=(-\Delta-\zeta)^{-1}$.
It is well known that $R_0(y,\zeta)=e^{i\sqrt\zeta|y|}/4\pi|y|$.
Therefore,
\be\la{glam}
g_{\lambda}(y)=\fr{e^{-\varkappa|\tilde y|-\varkappa_1\tilde y_1}}{4\pi|\tilde y|}
\ee
where $\varkappa_1:=\gamma|v|\lambda$. We choose $\Re\kappa>0$ for $\Re\lambda>0$.
Let us note that for $0<|v|<1$\be\la{kak}
 0<\Re\varkappa_1<\Re\varkappa,~~~~~~\Re\lambda>0
\ee
 Let us state the result which we have got above.
\begin{lemma} \la{cac}
i) The function $g_{\lambda}(y)$
decays exponentially in $y$ for $\Re\lambda>0$.
\\
ii) The formulas (\re{glam}) and  (\re{kappa}) imply that for every fixed $y$,
the  function $g_\lam(y)$ admits an analytic continuation in $\lam$
to the Riemann surface of the algebraic function $\sqrt{\lam^2+\mu^2}$ with the
branching points $\lambda=\pm i\mu$.
\end{lemma}
Thus, from (\re{eq1}) and (\re{F1}) we obtain the  representation
\beqn\la{Psi}
  \tilde\Psi_1&=&-G_{\lambda}^{11}\Psi_{01}-G_{\lambda}^{12}\Psi_{02}-
(G_{\lambda}^{12}\na\rho_1)\cdot\tilde Q+(G_{\lambda}^{11}\na\rho_2)\cdot\tilde Q\\
\nonumber
 \tilde\Psi_2&=&-G_{\lambda}^{21}\Psi_{01}-G_{\lambda}^{22}\Psi_{02}
-(G_{\lambda}^{22}\na\rho_1)\cdot\tilde Q+(G_{\lambda}^{21}\na\rho_2)\cdot\tilde Q
\eeqn
\noindent{\it Step ii)}
Now we proceed to the last two equations (\re{eq1}):
\be\la{lte}
-\lambda\tilde Q+B_v\tilde P=-Q_0,\quad\quad \langle\na\psi_{vj},
\tilde Q\cdot\na\rho_j\rangle-\langle\na\tilde\Psi_j,\rho_j\rangle-\lambda\tilde P=-P_0.
\ee
Let us rewrite equations (\re{Psi}) as
$\tilde\Psi_j=\tilde\Psi_{j}(\tilde Q)+\tilde\Psi_{j}(\Psi_{0})$, where
\be\la{PsiP}
  \tilde\Psi_1(\Psi_{0})=-G^{11}_{\lambda}\Psi_{01}-G^{12}_{\lambda}\Psi_{02},
  ~~\tilde\Psi_2(\Psi_{0})=-G^{21}_{\lambda}\Psi_{01}-G^{22}_{\lambda}\Psi_{02}
\ee
\be\la{PsiQ}
  \tilde\Psi_1(\tilde Q)=(-G^{12}_{\lambda}\na\rho_1
+G^{11}_{\lambda}\na\rho_2)\cdot\tilde Q,~~
  \tilde\Psi_2(\tilde Q)=(-G^{22}_{\lambda}\na\rho_1
+G^{21}_{\lambda}\na\rho_2)\cdot\tilde Q
\ee
Then
$\langle\na\tilde\Psi_j,\rho_j\rangle=\langle\na\tilde\Psi_{j}(\tilde Q),\rho_j\rangle+
\langle\na\tilde\Psi_{j}(\Psi_{0}),\rho_j\rangle$, and the last equation (\re{lte}) becomes
$$
\langle\na\psi_{vj},\tilde Q\cdot\na\rho_j\rangle-\langle\na\tilde\Psi_{j}(\tilde Q),
\rho_j\rangle-\lambda\tilde P=-P_0+\langle\na\tilde\Psi_{j}(\Psi_{0}),\rho_j\rangle=:-P_0-\Phi(\lam)
$$
where
\be\la{Philam}
\Phi(\lam)=\langle\tilde\Psi_{j}(\Psi_{0}),\na\rho_j\rangle
\ee
First we compute the term
$$
  \langle\na\psi_{vj},\tilde Q\cdot\na\rho_j\rangle
=\sum_{lj}\langle\na\psi_{vj},\tilde Q_l\pa_l\rho_j\rangle=
  \sum_{lj} \langle\na\psi_{vj},\pa_l\rho_j\rangle \tilde Q_l
$$
Applying the Fourier transform $F_{y\to k}$, we have by the Parseval
identity and (\re{pp}) that
\beqn\la{Lij}
  \sum_j\langle\pa_i\psi_{vj},\pa_l\rho_j\rangle\!\!&=&\!\!
  \sum_j\langle -ik_i\hat\psi_{vj},-ik_l\hat\rho_j\rangle
  =\int k_ik_l(\hat\psi_{v1}\cdot\hat\rho_1+\hat\psi_{v2}\cdot\hat\rho_2)dk\\
  \nonumber\\
  \nonumber
 \!\! &=&\!\!-\!\int k_ik_lm\frac{\beta\hat\rho_1\cdot\hat\rho_1+\beta\hat\rho_2\cdot\hat\rho_2}
  {k^2+m^2-(|v|k_1)^2}dk=-\!\int\frac{{ k_ik_l\cal B}(k)dk}{k^2+m^2-(|v|k_1)^2}=:-L_{il}
\eeqn
As the result, $\langle\na\psi_{vj},\tilde Q\cdot\na\rho_j\rangle=-L\tilde Q$, where $L$
is the $3\times3$ matrix with the matrix elements $L_{il}$.
The matrix $L$ is diagonal and positive defined by (\re{W}).\\
Now let us compute the term 
$-\langle\na\tilde\Psi_j(\tilde Q),\rho_j\rangle=\langle\tilde\Psi_j(\tilde Q),\na\rho_j\rangle$.
One has
$$
  \langle\tilde\Psi_j(\tilde Q),\pa_i\rho_j\rangle
  \!=\!\!\sum\limits_l\!\Big(\langle -G^{12}_{\lambda}\pa_l\rho_1
  + G^{11}_{\lambda}\pa_l\rho_2,\pa_i\rho_1\rangle
  -\langle G^{22}_{\lambda}\pa_l\rho_1-
  G^{21}_{\lambda}\pa_l\rho_2,\pa_i\rho_2\rangle\Big)\tilde Q_l
  \!=\!\!\sum\limits_l H_{il}(\lambda)\tilde Q_l
$$
and  by the Parseval identity and (\ref{aij})-(\ref{al-pr}) we have
\beqn \nonumber
H_{il}(\lambda):
&=&\langle -G^{12}_{\lambda}\pa_l\rho_1+G^{11}_{\lambda}\pa_l\rho_2,\pa_i\rho_1\rangle
  -\langle G^{22}_{\lambda}\pa_l\rho_1-G^{21}_{\lambda}\pa_l\rho_2,\pa_i\rho_2\rangle\\
  \nonumber
  &=&\langle [(\beta m-\al_2k_2)\hat\rho_1-(i\al_1k_1+i\al_3k_3+i|v|k_1+\lambda)\hat\rho_2]
  \hat g_{\lambda}k_l,k_i\hat\rho_1\rangle\\
 \la{Cij}
  &+&\langle[(i\al_1k_1+i\al_3k_3+i|v|k_1+\lambda)\hat\rho_1+(\beta m-\al_2k_2)\hat\rho_2]
  \hat g_{\lambda}k_l,k_i\hat\rho_2\rangle\\
  \nonumber
  &=&\int k_ik_lm\frac{\beta\hat\rho_1\cdot\hat\rho_1+\beta\hat\rho_2\cdot\hat\rho_2}
  {k^2+m^2-(|v|k_1-i\lambda)^2}dk=\int \frac{k_ik_l{\cal B}(k)dk}{k^2+m^2-(|v|k_1-i\lambda)^2}
\eeqn
The matrix $H$ is well defined for $\Re\lambda>0$ since the denominator does not vanish.
The matrix $H$ is diagonal. Indeed, if $i\ne l$, then at least one of these indices
is not equal to one, and the integrand in (\re{Lij}) is odd with respect to the corresponding
variable. Thus, $H_{il}=0$.
As the result,
$\langle\tilde\Psi_j(\tilde Q),\na\rho_j\rangle=H\tilde Q$, where $H$ is the
matrix with matrix elements $H_{il}$.
Finally the equations (\re{lte}) become
\be\la{Mlam}
  M(\lambda)\left(
  \begin{array}{c}
  \tilde Q \\ \tilde P
  \end{array}
  \right)=\left(
  \begin{array}{c}
  Q_0 \\ P_0+\Phi(\lam)
  \end{array}
  \right),\,\,{\rm where}\,\,M(\lambda)=\left(
  \begin{array}{cc}
  \lambda E  & -B_v \\
  L-H(\lambda) & \lambda E
  \end{array}
  \right)
\ee
Assume for a moment that the matrix $M(\lambda)$ is invertible (later
we will prove this). Then we obtain
\be\la{QP1}
  \left(
  \begin{array}{c}
  \tilde Q \\ \tilde P
  \end{array}
  \right)=M^{-1}(\lambda)\left(
  \begin{array}{c}
  Q_0 \\ P_0+\Phi(\lam)
  \end{array}
  \right),~~~~~~~~~\Re\lambda>0
\ee
Finally, formula (\re{QP1}) and formulas (\re{Psi}), where $\tilde Q$ is expressed from (\re{QP1}),
give the expression of the resolvent $R(\lambda)=(A-\lambda)^{-1}$, $\Re\lambda>0$.
\begin{lemma}\la{cmf}
  The matrix function $M(\lambda)$ (respectively, $M^{-1}(\lambda)$),
  $\Re\lambda>0$ admits
  an analytic (respectively meromorphic) continuation to the Riemann surface of the function\\
  $\sqrt{\mu^2+\lambda^2}$, $\lambda\in\C$.
\end{lemma}
\begin{proof}
  The analytic continuation of  $M(\lambda)$,  exists by Lemma \re{cres} ii)
  and the
  convolution expressions in (\re{Cij}) by  (\re{ro}). The inverse matrix is then
  meromorphic since it exists for large $\Re\lambda$. The latter follows from (\re{Mlam})
  since $H(\lambda)\to 0$, $\Re\lambda\to\infty$, by  (\re{Cij}).
\end{proof}
\setcounter{equation}{0}
\section{Regularity on  imaginary axis}
Let us describe  the continuous spectrum of the operator $A=A_{v,v}$ 
on the imaginary axis. By definition, the continuous spectrum corresponds to
$\om\in\R$, such that the resolvent $R(i\om+0)$ is not a bounded operator
in ${\cal E}$. By the formulas (\re{Psi}), this is the case when the Green function
$G_\lambda(y-y')$ loses the exponential decay. This is equivalent to the condition
$\Re\varkappa=0$. Thus, $i\om$ belongs to the continuous spectrum if
$|\om|\ge\mu=m\sqrt{1-v^2}$.
By Lemma \re{cmf}, the limit matrix
\be\la{M}
 M(i\om):=M(i\om+0)=\left(
  \begin{array}{cc}
  i\om E        & -B_v      \\
  L  -H(i\om+0) & i\om E
  \end{array}
  \right), \quad\om\in\R
\ee
exists, and its entries are continuous functions of $\om\in\R$, smooth for $|\om|<\mu$
and $|\om|>\mu$. Recall that the point $\lambda=0$ belongs to the discrete spectrum of
the operator $A$ by Lemma \re{ceig}, hence $M(i\om+0)$
(probably)  also is not invertible at $\om=0$.
\begin{prop}\la{regi} (cf. \cite [Proposition 15.1] {ikv})
  Let $\rho$ satisfy the conditions (\ref{ro})- (\ref{W}), and $|v|<1$.
  Then the limit matrix $M(i\om+0)$ is invertible for $\om\ne 0$, $\om\in\R$.
\end{prop}
\begin{cor}\la{creg}
The matrix $M^{-1}(i\om)$ is smooth in $\om\in\R$ outside three points $\om=0,\pm \mu$.
\end{cor}
\setcounter{equation}{0}
\section{Singular spectral points}
The components $Q(t)$ and $P(t)$ are given by the Fourier integral
\be\la{QP1i}
  \left(
  \begin{array}{c}
  Q(t) \\ P(t)
  \end{array}
  \right)=\ds\fr 1{2\pi}\int e^{i\om t}M^{-1}(i\om+0)\left(
  \begin{array}{c}
  Q_0 \\ P_0+\Phi(i\om)
  \end{array}
  \right)d\om
\ee
if it converges in the sense of distributions. Corollary \re{creg} alone is not sufficient for the proof
of the convergence and decay of the integral. Namely, we need an additional information about behavior
of the matrix $M^{-1}(i\om)$ near its singular points  $\om=0,\pm \mu$, and asymptotics at $|\om|\to\infty$.
We will analyze all the points separately.
\medskip

{\bf I.}
First we consider the points $\pm \mu$.
\begin{lemma} \la{Pui}
  The matrix  $M^{-1}(i\om)$ admits the  asymptotics
  in a vicinity of $\pm \mu$: 
  \be\la{mom}
    M^{-1}(i\om)=C^\pm+{\cal O}((\om\mp\mu)^{\fr 12}),~~
    \pa_{\om} M^{-1}(i\om)={\cal O}((\om\mp\mu)^{-\fr 12}),~~
\pa_{\om}^2 M^{-1}(i\om)={\cal O}((\om\mp\mu)^{-\fr 32})
\ee
\end{lemma}
\begin{proof}
  It suffices to prove  similar asymptotics  for $M(i\om)$. Then  (\re{mom}) holds also
  for $M^{-1}(i\om)$,
  since the matrices $M(\pm i\mu)$ are invertible. The asymptotics for $M(i\om)$ holds by
  the convolution representation  (\re{Cij})
 \be\la{Cjjj}
  H_{jj}(\lambda)=\langle mg_{\lam}\beta*\pa_j\rho_1,\pa_j\rho_1\rangle
+\langle mg_{\lam}\beta*\pa_j\rho_2,\pa_j\rho_2\rangle,\quad j=1,2,3
\ee
since  $g_{\lambda}$ admits the corresponding asymptotics by the
formula (\re{glam}). Namely
$$
g_{\lambda}(y)=\frac 1{4\pi|\tilde y|}+r_{\pm}(\lam,y),\quad\lam\to\pm i\mu,\quad\Re\lam>0
$$
where
$$
r_{\pm}(\lam,y)\!=\!{\cal O}((\lam\mp i\mu)^{\fr 12}),~~
\pa_{\lam}r_{\pm}(\lam,y)\!=\!{\cal O}((\lam\mp i\mu)^{-\fr 12}),~~
\pa_{\lam}^2r_{\pm}(\lam,y)\!=\!{\cal O}((1\!+\!|y|)(\lam\mp i\mu)^{-\fr 32})
$$
The condition (\ref{ro}) provides the
convergence of all integrals arising in $\pa_{\lam}^{k}H_{jj}$.
\end{proof}
{\bf II.}
Second, we study the asymptotic behavior of $M^{-1}(\lambda )$ at infinity. 
\begin{lemma}\la{162}
There exist a matrix $D_{0}$ and a matrix-function $D_{1}(\om )$,  such that
\be\la{Minv}
M^{-1}(i\om)=\fr
 {D_0}\om +D_1(\om ),\quad|\om|\ge\mu+1,\quad\om \in \R
\ee
where, for  $k=0,1,2$
\be\label{min}
|\pa_\om^k D_1(\om )|\leq \fr{C(k)}{|\om |^2},\quad |\om|\ge\mu+1,\quad\om \in \R
\ee
\end{lemma}
\begin{proof}
The structure (\ref{M}) of the matrix $M(i\om )$ provides that it suffices to prove the
following estimate for the elements of the matrix $H(i\om ):=H(i\om+0)$:
\be \la{minC}
  |\pa _\lam^k H_{jj}(\lam)|\leq C(k),\quad\lam\in\C,
\quad|\lam|\ge\mu+1,\quad j=1,2,3,\quad k=0,1,2
\ee
The estimate  (\ref{minC})  follows from the representation
(\re{Cjjj}) and the bounds
$$
|g_{\lambda}(y)|\le\frac{C_1}{|y|},\quad 
|\pa_{\lambda}g_{\lambda}(y)|\le \frac{C_2}{|y|}+C_3,\quad
|\pa_{\lambda}^2g_{\lambda}(y)|\le \frac{C_4}{|y|}+C_5|y|,\quad\Re\lam>0
$$
\end{proof}
{\bf III.}
Finally, we consider the point  $\om=0$ which is
an isolated pole of a finite degree by Lemma \re{cmf}.
In Appendix B we prove that the matrix $M^{-1}(i\om)$ can be written
in the form
\be\la{calM}
M^{-1}(i\om)=\left(\ba{ll}
\fr{1}{\om}{\cal M}_{11}(\om)  & \fr{1}{\om^2}{\cal M}_{12}(\om) \\\\
{\cal M}_{21}(\om)  & \fr{1}{\om}{\cal M}_{22}(\om)
\ea\right)
\ee
where ${\cal M}_{ij}(\om)$, $i,j=1,2$ are the diagonal matrices,
smooth for the  $\om\in(-\mu,\mu)$. Moreover,
\be\la{cMM}
{\cal M}_{22}={\cal M}_{11},\quad
{\cal M}_{11}=i{\cal M}_{12}B_{\nu}^{-1}
\ee

\setcounter{equation}{0}
\section{Transversal decay for the linearized equation }
Here we  prove Proposition 
(\re{lindecay}).
\subsection{Decay of  vector components}
First, we  establish 
the decay (\re{frozenest}) for the components $Q(t)$ and $P(t)$.
\begin{lemma}\la{171} 
Let $X_0\in {\cal Z}_{v}\cap {\cal E}_{\nu}$. Then $Q(t)$, $P(t)$ are continuous and
\be\la{decQP}
|Q(t)|+|P(t)|\le C_{\nu}(\rho,\tilde v)(1+|t|)^{-3/2},
~~~~~~~t\ge 0.
\ee
\end{lemma}
{\it Proof}.
The expansions (\re{mom}), (\re{Minv})  and (\re{calM}) imply
the convergence of the Fourier integral (\re{QP1i})
in the sense of distributions to a continuous function
of $t\ge 0$. Let us prove (\re{decQP}).
First let us note that the condition  $X_0\in {\cal Z}_{v}$ implies that
the whole trajectory $X(t)$ lies in  ${\cal Z}_{v}$. This follows
from the invariance of the space ${\cal Z}_{v}$ under the generator
$A_{v,v}$ (cf. Remark \re{rZ}). If $X_0\not\in{\cal Z}_{v}$, then
the components  $Q(t)$ and $P(t)$ may contain non-decaying terms
which correspond to the singular point $\om=0$ since the linearized
dynamics admits the secular solutions without decay, see (\re{secs}).
We will show that  the symplectic orthogonality condition leads to (\re{decQP}).
Let us split the Fourier integral (\re{QP1i}) into three terms using the partition of unity
$\zeta_1(\om)+\zeta_2(\om)+\zeta_3(\om)=1$, $\om\in\R$:
$$
\left(\!\begin{array}{c} Q(t) \\ P(t)\end{array}\!\right)=
\fr 1{2\pi}\int e^{i\om t}(\zeta_1(\om)+\zeta_2(\om)+\zeta_3(\om))M^{-1}(i\om+0)
\left(\!\begin{array}{c}Q_0 \\ P_0+\Phi(i\om)\end{array}\!\right)d\om
=\sum\limits_{j=1}^3 I_j(t)
$$
where the functions $\zeta_j(\om)\in C^\infty(\R)$ are supported by
\be\la{zsup}
\left.\begin{array}{rcl}
\supp \zeta_1&\subset& \{\om\in\R:\ve_0/2<|\om|<\mu+2\}
\\
\\
\supp \zeta_2&\subset& \{\om\in\R:|\om|>\mu+1\}
\\
\\
\supp \zeta_3&\subset& \{\om\in\R:|\om|<\ve_0\}
\end{array}\right|
\ee
i) Let us represent  $I_j(t)$, $j=1,2$ as
\beqn\nonumber
I_j(t)&=&\frac{1}{2\pi}\int e^{i\om t}\zeta_j(\om)\Big[M^{-1}(i\om+0)
\left(\begin{array}{c} Q_0 \\ P_0 \end{array}\right)+M^{-1}(i\om+0)
\left(\!\begin{array}{c} 0 \\ \Phi(i\om) \end{array}\!\right)\Big]d\om\\
\la{I2}
&=&s_{j}(t)\left(\begin{array}{c} Q_0 \\ P_0 \end{array}\right)
+s_j(t)*\left(\begin{array}{c} 0 \\ f(t) \end{array}\right)
\eeqn
where 
$$
s_j(t)=\Lambda^{-1}\zeta_j(\om)M^{-1}(i\om+0),\quad
f(t)=\Lambda^{-1}\Phi(i\om)
$$
By (\ref{PsiP})
$$
 \tilde\Psi_1(\Psi_0)=-\Lambda \Re W_v^+(t)(\Psi_{10}+i\Psi_{20}),\quad
 \tilde\Psi_2(\Psi_0)=-\Lambda \Im W_v^+(t)(\Psi_{10}+i\Psi_{20})
$$
where $W_v^+(t)$ is  the dynamical group of the  equation
\be\la{DEf1}
  \dot\Psi(x,t)=[\al\cdot\nabla+i\beta m+v\cdot\nabla]\Psi(x,t)
\ee
Evidently,  for  the group the $W_v^+(t)$  the  bound
(\re{Wv}) obtained in Lemma \re{Vain} for the group $W_v(t)$ also holds.
Hence,  (\ref{Philam}) implies
\be\la{f-est}
|f(t)|=|\Re\langle W_v^+(t)(\Psi_{10}+i\Psi_{20}),\nabla\rho\rangle|
\le C_{\nu}(\rho,v)(1+t)^{-3/2}
\ee
Further, the function $s_1(t)$ decays as $(1+|t|)^{-3/2}$ 
by asymptotics (\re{mom}), and  the function $s_2(t)$ decays
as  $(1+|t|)^{-2}$ due to Proposition \re{162}.
Hence,  formula (\ref{I2}) implies the  decay $(1+|t|)^{-3/2}$ for
$I_1(t)$ and $I_2(t)$.\\
iii) Finally, the function $I_3(t)$ decays as $t^{-\infty}$  if $Z_0\in{\cal Z}_{v}$.
It follows from  next lemma
\begin{lemma}\la{vdec}
If $Z_0\in{\cal Z}_{v}$ then
\be\la{tQP}
  \left(\begin{array}{c} \tilde Q(i\om) \\ \tilde P(i\om)\end{array}\right)=
  M^{-1}(i\om)\left(\begin{array}{c} Q_0 \\ P_0+\Phi(i\om)\end{array}\right)
  \in C^{\infty}(-\mu,\mu)
\ee
\end{lemma}
\begin{proof}
In Appendix C we prove that the symplectic orthogonality conditions (\ref{orth}) at $t=0$
 imply 
\be\la{soin}
P_0+\Phi(0)=0,\quad B_v^{-1}Q_0+\Phi'(0)=0
\ee
Then
$$
P_0+\Phi(i\om)=\Phi(i\om)-\Phi(0)=i\om\Upsilon_1(\om)
$$
$$
B_v^{-1}Q_0+\Upsilon_1(\om)=\frac{\Phi(i\om)-\Phi(0)}{i\om}-\Phi'(0)
=i\om\Upsilon_2(\om)
$$
where $\Upsilon_j(\om)\in C^{\infty}(-\mu,\mu)$, since
$\tilde\Psi_j(\Psi_0)\in C^{\infty}(-\mu,\mu)$,  by (\ref{PsiP}) and (\ref{Gelam}).
Therefore, representations (\ref{calM})-(\ref{cMM}) imply
\beqn\nonumber
\tilde P(i\om)&=&{\cal M}_{21}(\om)Q_0+i{\cal
  M}_{22}(\om)\Upsilon_1(\om)\in C^{\infty}(-\mu,\mu)\\
\nonumber
\tilde Q(i\om)&=&\frac{1}{\om}{\cal M}_{11}(\om)Q_0
+\frac{i}{\om}{\cal M}_{12}(\om)\Upsilon_1(\om)
=\frac{i}{\om}{\cal M}_{12}(B_v^{-1}Q_0+\Upsilon_1(\om))\\
\nonumber
&=&-{\cal M}_{12}\Upsilon_2(\om)\in C^{\infty}(-\mu,\mu)
\eeqn
\end{proof}
\setcounter{equation}{0}
\subsection{Decay of fields}
Now we prove the decay  of the field components $\Psi_1(x,t),\Psi_2(x,t)$ corresponding
to (\re{frozenest}). The first  two equations of (\re{Avv}) may be written as one equation:
\be\la{DE}
  \dot\Psi(x,t)=[-\al\cdot\nabla-i\beta m+v\cdot\nabla]\Psi(x,t)-iQ(t)\cdot\nabla\rho,
  \quad x\in\R^3,\quad t\in\R
\ee
where $\Psi(t)=\Psi_1(\cdot,t)+i\Psi_2(\cdot,t))$.
Applying  the Duhamel representation, we obtain 
\be\la{Duh}
\Psi(t)=W_v(t)\Psi_0-\int_0^tW_v(t-s) Q(s)\cdot\na\rho~ds,\quad t\ge 0
\ee
where $W_v(t)$  the dynamical group (propagator) of the "modified"
free Dirac equation
\be\la{DEf}
  \dot\Psi(x,t)=[-\al\cdot\nabla-i\beta m+v\cdot\nabla]\Psi(x,t)
\ee
Lemma \re{Vain} on
the weighted decay for the group $W_v(t)$,
the decay of $Q$ from (\re{decQP}), and representation (\ref{Duh}) yield 
\be\la{lins}
\Vert \Psi(t)\Vert_{-\nu}
\le C_{\nu}(\rho,\tilde v)\Vert\Psi_0\Vert_{\nu}(1+|t|)^{-3/2},~~~~t\ge 0
\ee
for any 
$\Psi_0\in{\cal Z}_{v}\cap {\cal E}_{\nu}$.
It completes the proof of Proposition \re{lindecay}.

\section{Weighted decay for free Dirac equation}
\begin{lemma}\la{Vain}
For any $\Phi \in L^2_\nu$ with $\nu>3/2$ the bound holds
\be\la{Wv}
 \Vert W_v(t)\Phi\Vert_{-\nu}\le\fr{C_{\nu}(v)\Vert\Phi\Vert_{\nu}}{(1+|t|)^{3/2}},\quad t\ge 0
\ee
\end{lemma}
\begin{proof}
{\it Step i)}  Note, that
$$
\left(\partial_t+\al\cdot\nabla+i\beta m+v\cdot\nabla\right)
\left(\partial_t-\al\cdot\nabla-i\beta m+v\cdot\nabla\right)=
(\partial^2_t +2\partial_tv\cdot\nabla-\Delta(1-v^2)+m^2)
$$
Hence the integral kernel $W_v(x-y,t)$ of the  operator $W_v(t)$  has the form
\be\la{23}
W_v(z,t)=\left(\partial_t+\al\cdot\nabla+i\beta m+v\cdot\nabla\right)G_v(z,t),
\ee
where  $G_v(z,t)$ is a fundamental solution of the "modified" Klein-Gordon operator
$$
(\partial^2_t +2\partial_tv\cdot\nabla-\Delta(1-v^2)+m^2)G_v(z,t)=\delta(z)\delta(t)
$$
Let $G_v(t)$, $t\ge 0$ be the  operator with the integral kernel $G_v(x-y,t)$.
It is easy to see that
$$
[G_v(t)\Phi](x)=[G_0(t)\Phi](x-vt),\quad x\in\R^3,\quad t\ge 0
$$
Then
$$
G_v(z,t)=G_0(z-vt,t)=\frac{\delta (t-|z-vt|)}{4\pi t}-\fr{m}{4\pi}\fr{\theta(t-|z-vt|)J_{1}
(m\sqrt{t^{2}-|z-vt|^{2}})}{\sqrt{t^{2}-|z-vt|^{2}}},~~~~t>  0
$$
where $J_{1}$ is the Bessel function of order 1,
and $\theta$ is the Heavyside function.
Let us fix an arbitrary $\ve\in(|v|,1)$.
Well known asymptotics of the Bessel function imply that
\be\la{vWv}
 |\pa_{t}G_{v}(z,t)|,\,|\pa_{z_j}G_{v}(z,t)| \le C(\ve)(1+t)^{-3/2},\quad
|z-vt|\le\ve t,~~t\ge 1,~~j=1,2,3
\ee
\textit{Step ii)}
Now we consider an arbitrary $t\ge 1$. Denote $\ve_1=\ve-|v|$. 
We split the  function $\Phi$ in two terms,
$\Phi=\Phi_{1,t}+\Phi_{2,t}$ such that
\begin{equation} \label{FFn}
  \Vert \Phi_{1,t}\Vert _{L^2_\nu}+\Vert \Phi_{2,t}
  \Vert _{L^2_\nu}\le C\Vert \Phi\Vert _{L^2_\nu},\quad t\ge 1
\end{equation}
and
\be\label{F}
  \Phi_{1,t}(x)=0~~\mbox{for}~|x|>\frac{\ve_1 t}{2},
  ~~~~~~~\mbox{and}~~~~~~~~~~
  \Phi_{2,t}(x)=0~~\mbox{for}~|x|<\frac{\ve_1 t}{4}
  \ee
The estimate (\re{lins}) for $W_v(t)\Phi_{2,t}$ follows by
charge conservation for Dirac equation,
(\ref {F})  and (\ref{FFn}):
\be\label{zxc}
  \Vert W_v(t)\Phi_{2,t}\Vert _{L^2_{-\nu}}\le
  \Vert W _v(t)\Phi_{2,t}\Vert _{L^2_0}=\Vert\Phi_{2,t}\Vert _{L^2_0}
  \leq \fr{C(\ve)\Vert \Phi_{2,t}\Vert _{L^2_\nu}}{(1+t)^{\nu}}
  \le \fr{C_{1}(\ve)\Vert\Phi\Vert _{L^2_\nu}}{(1+t)^{3/2}},\quad t\geq 1
\ee
since $\nu>3/2$.
\\
\textit{Step iii)}
Next we consider $W _v(t)\Phi_{1,t}$.
Now we split the operator $W_v(t)$ in two terms:
$$
 W_v(t)=(1-\zeta )W_v(t)+\zeta W_v(t),~~~~~t\ge 1
$$
where $\zeta $ is the operator of multiplication by the function $\zeta ({|x|}/{t})$
such that $\zeta =\zeta (s)\in C_{0}^{\infty }(\R)$,
$\zeta (s)=1$ for $|s|<\ve_1/4$, $\zeta (s)=0$ for $ |s|>\ve_1/2$.
Since $1-\zeta ({|x|}/{t})=0$ for $|x|<\ve_1 t/4$, then
applying  the charge conservation and (\ref{FFn}), we have for $t\ge 1$
\be \label{zaz}
  \Vert(1-\zeta )W_v(t)\Phi_{1,t}\Vert_{L^2_{-\nu}}
  \!\le\!\fr{C(\ve)\Vert W_v(t)\Phi_{1,t}\Vert_{L^2_0}}{(1+t)^{\nu}}
\!=\!\fr{C(\ve)\Vert\Phi_{1,t}\Vert_{L^2_0}}{(1+t)^{\nu}}
\!\le\!\fr{C_{1}(\ve)\Vert\Phi_{1,t}\Vert_{L^2_\nu}}{(1+t)^{\nu}}
\!\le\!\fr{C_{2}(\ve)\Vert\Phi\Vert_{L^2_\nu}}{(1+t)^{3/2}} 
\ee
since $\nu>3/2$.
\\
\textit{Step iv)}
It remains to estimate $\zeta W_v(t)\Phi_{1,t}$.
Let $\chi_{t}$ be the characteristic function of the ball $ |x|\le \ve_1 t/2$.
We will use the same notation for the operator of multiplication by this characteristic
function. By (\ref{F}), we have
\be\la{zeze}
  \zeta W_v(t)\Phi_{1,t}=\zeta W_v(t)\chi_{t}\Phi
\ee
The matrix kernel of the operator $\zeta W_v(t)\chi_{t}$ is equal to
$$
  W'_v(x-y,t)=\zeta ({|x|}/{t})W_v(x-y,t)\chi_{t}(y)
$$
Since $\zeta ({|x|}/{t})=0$ for $|x|>\ve_1 t/2 $ and $\chi_{t}(y)=0$ 
for $|y|>\ve_1 t/2$. Therefore,  $W'_v(x-y,t)=0$ for $|x-y|>\ve_1 t$.
On the other hand,  $|x-y|\le\ve_1 t$ implies $|x-y-vt|\le\ve t$, 
since $\ve_1+|v|=\ve$ by definition of $\ve_1$. Hence, equality (\ref{23})
and bounds (\ref{vWv})  yield
\begin{equation}\label{qaz}
  |W'_v(x-y,t)|\le C(1+t)^{-3/2},\quad \quad t\ge 1
\end{equation}
The norm of the operator
$\zeta W_v(t)\chi_{t}: L^2_{\nu}\rightarrow L^2_{-\nu}$
is equivalent to the norm of the operator
\be\la{1234}
  \langle x\rangle^{-\nu}\zeta W_v(t)\chi_{t}(y)
\langle y\rangle^{-\nu}:L^2_0 \rightarrow L^2_0
\ee
Therefore,  (\ref{qaz}) implies that operator (\ref{1234}) is Hilbert-Schmidt
operator since $\nu >3/2,$ and its Hilbert-Schmidt norm does not
exceed $C(1+t)^{-3/2}$. Hence, by (\ref{zeze}) and (\ref{FFn}) 
\begin{equation} \label{HS}
  \Vert\zeta W_v(t)\Phi_{1,t}\Vert_{L^2_{-\nu}}
  \le C(1+t)^{-3/2}\Vert\Phi\Vert_{L^2_\nu},
\quad t\ge 1
\end{equation}
Finally, the estimates (\ref {HS}), (\ref{zaz}) and (\ref{zxc}) imply (\ref{lins}).
\end{proof}

\appendix
\section{Computing $\Om(\tau_i,\tau_j)$}
Here we  justify the formulas  (\re{Omega})-(\re{alpha}) for the matrix $\Om$.\\
1) First, the  Parseval identity implies
$$
   \Om(\tau_j,\tau_l)=\langle\pa_j\psi_{v1},\pa_l\psi_{v2}\rangle-
   \langle\pa_j\psi_{v2},\pa_l\psi_{v1}\rangle\!=\int\!\! k_jk_l~dk(\hat\psi_{v1}\cdot\hat\psi_{v2}
   -\hat\psi_{v2}\cdot\hat\psi_{v1})= 0,\quad j,l=1,2,3
$$
since the integrand is odd function.
\medskip\\
2) Second, we consider 
\be\la{j2}
   \Om(\tau_{j+3},\tau_{l+3})=\langle\pa_{v_j}\psi_{v1},\pa_{v_l}\psi_{v2}\rangle-
   \langle\pa_{v_j}\psi_{v2},\pa_{v_l}\psi_{v1}\rangle
\ee
Let us derive the formulas for $\psi_{v1}$ and $\psi_{v2}$.
The first equation of \eqref{stfch} implies
$$
[(v_j\pa_j)^2-\Delta+m^2]\psi_{v}
=[iv_j\pa_j+i\al_j\pa_j-\beta m]\rho_1
$$
Hence
$$
[(v_j\pa_j)^2-\Delta+m^2]\psi_{v1}=
-[v_j\pa_j+\al_1\pa_1+\al_3\pa_3]\rho_2 - [\tilde\al_2\pa_2+\beta m]\rho_1
$$
$$
[(v_j\pa_j)^2-\Delta+m^2]\psi_{v2}=
[v_j\pa_j+\al_1\pa_1+\al_3\pa_3]\rho_1-[\tilde\al_2\pa_2+\beta m]\rho_2
$$
Applying the Fourier transform, we obtain
\be\la{system}
\left.\begin{array}{l}  
\hat\psi_{v1}=\ds\frac{[ivk+i\al_1k_1+i\al_3k_3]\hat\rho_2+
[\al_2k_2-\beta m]\hat\rho_1}{-(vk)^2+k^2+m^2}\\
\hat\psi_{v2}=\ds\frac{-[ivk+i\al_1k_1+i\al_3k_3]\hat\rho_1+[\al_2k_2-\beta m]\hat\rho_2}
{-(vk)^2+k^2+m^2}
\end{array}\right|
\ee
Differentiating, we get
\be\la{derv}
\left.\begin{array}{l} 
   \pa_{v_j}\hat\psi_{v1}=\ds\fr{ik_j\hat\rho_2}{-(vk)^2+k^2+m^2}
   +\fr{2k_jvk\hat\psi_{v1}}{-(vk)^2+k^2+m^2}\\
    \pa_{v_l}\hat\psi_{v2}=\ds\fr{-ik_l\hat\rho_1}{-(vk)^2+k^2+m^2}
   +\fr{2k_lvk\hat\psi_{v2}}{-(vk)^2+k^2+m^2}
\end{array}\right|
\ee
Hence, (\re{j2}) implies
\beqn\nonumber
  \Om(\tau_{j+3},\tau_{l+3})\!\!&=&\!\!\int\!\frac
  {k_jk_l[\hat\rho_1\cdot\hat\rho_2-\hat\rho_2\cdot\hat\rho_1]dk}{(k^2+m^2-(vk)^2)^2}
  +\int\!\frac{4k_jk_l(vk)^2
  [\hat\psi_{v1}\cdot\hat\psi_{v2}-\hat\psi_{v2}\cdot\hat\psi_{v1}]dk}
  {(k^2+m^2-(vk)^2)^2}\\
  \nonumber
  \!\!&+&\!\!\int\!\frac{2ik_jk_lvk[\hat\rho_2\cdot\hat\psi_{v2}+\hat\psi_{v2}\cdot\hat\rho_2
  +\hat\rho_1\cdot\hat\psi_{v1}+\hat\psi_{v1}\cdot\hat\rho_1]}{(k^2+m^2-(vk)^2)^2}=0
\eeqn
since  all integrands are odd functions.
\medskip\\
3) Finally, (\ref{derv}) implies
\beqn \la{p}
   \Om(\tau_j,\tau_{l+3})\!\!\!&=&\!\!\!-\langle\pa_{j}\psi_{v1},\pa_{v_l}\psi_{v2}\rangle+
   \langle\pa_{j}\psi_{v2},\pa_{v_l}\psi_{v1}\rangle+e_j\cdot \pa_{v_l}p_v\\
   \nonumber   
   \!\!\!&=&\!\!\!\int \frac{ik_j\hat\psi_{v1}\cdot[-ik_l\hat\rho_1+2k_lvk\hat\psi_{v2})]
   -ik_j\hat\psi_{v2}\cdot[ik_l\hat\rho_2+2k_lvk\hat\psi_{v1})]}
   {k^2+m^2-(vk)^2}dk +e_j\cdot \pa_{v_l}p_v\\
\nonumber  
\!\!\! &=& \!\!\!\!\int\!\! k_jk_l\frac{-[\hat\psi_{v1}\cdot
\hat\rho_1+\hat\psi_{v2}\cdot\hat\rho_2]
+2ivk[\hat\psi_{v1}\cdot\hat\psi_{v2}-\hat\psi_{v2}\cdot\hat\psi_{v1}]}
{k^2+m^2-(vk)^2}dk +e_j\cdot \pa_{v_l}p_v
\eeqn
Recall, that  $\rho_j(x)$ are even, then $\hat\rho_j(k)$ are
real. Hence (\ref{aij})-(\ref{al-pr}) and (\re{system}) imply
\beqn\la{pp}
\!\!\!\!&&\!\!\!\!\!\!\!\!\!\!\!\!\!\!\!\!
(k^2+m^2-(vk)^2)(\hat\psi_{v1}\cdot\hat\rho_1+\hat\psi_{v2}\cdot\hat\rho_2)
=[\al_2k_2-\beta m]\hat\rho_1\cdot\hat\rho_1
+[\al_2k_2-\beta m]\hat\rho_2\cdot\hat\rho_\\
\nonumber
\!\!\!\!&&+[ivk+i\al_1k_1+i\al_3k_3]\hat\rho_2\cdot\hat\rho_1
-[ivk+i\al_1k_1+i\al_3k_3]\hat\rho_1\cdot\hat\rho_2=-{\cal B}\hat\rho\cdot\hat\rho
\eeqn
\beqn\la{ppp}
\!\!\!\!\!&&(k^2+m^2-(vk)^2)^2(\hat\psi_{v1}\cdot\hat\psi_{v2} -\hat\psi_{v2}\cdot\hat\psi_{v1})
=2i(k^2+m^2-(vk)^2)^2\Im(\hat\psi_{v1}\cdot\hat\psi_{v2})\\
\nonumber
\!\!\!\!\!&&=-2\beta m\hat\rho_1\cdot[ivk+i\al_1k_1+i\al_3k_3]\hat\rho_1
-2[ivk+i\al_1k_1+i\al_3k_3]\hat\rho_2\cdot\beta m\hat\rho_2
=-2ivk{\cal B}\hat\rho\cdot\hat\rho
\eeqn
Substituting (\re{pp}) and (\re{ppp}) into the right hand site of (\re{p}), we obtain
$$
\Om(\tau_j,\tau_{l+3})=\int k_jk_l\Big(\frac{{\cal B}(k)}{(k^2+m^2-(vk)^2)^2} 
+\fr{4(vk)^2{\cal B}(k)}{(k^2+m^2-(vk)^2)^3}\Big)dk+e_j\cdot \pa_{v_l}p_v
$$
that correspond to (\re{Omega}) - (\re{alpha}).
\section{Computing $M^{-1}(i\om)$}
Here we derive formula (\ref{calM}). Denote $F(\om):=-L+H(i\om+0)$ which is  diagonal.
Then by (\re{M}) for $\om\in\R$ we obtain
\be\la{detM}
{\rm det}\,M(i\om)\!=\!{\rm det}\left(\!\!\ba{ll}
i\om E  & -B_v \\
-F(\om) & i\om E
\ea\!\!\right)
\!=-\Big(\om^2+\fr{F_{11}(\om)}{\gamma^3}\Big)
\Big(\om^2+\fr{F_{22}(\om)}{\gamma}\Big)\Big(\om^2+\fr{F_{33}(\om)}{\gamma}\Big)
\ee
where
\be\la{Fjj}
F_{jj}(\om)=\int k_j^2{\cal B}dk\left(\fr1{m^2+k^2-(|v|k_1+\om)^2}-
\fr1{m^2+k^2-(|v|k_1)^2}\right),\,\,\,j=1,2,3
\ee
Formula (\re{detM})
is obvious since both matrices $F(\om)$ and $B_v$ are diagonal, hence the matrix
$M(i\om)$ is equivalent to three independent  matrices $2\times2$.
Namely, let us transpose the columns and rows of the matrix $M(i\om)$
in the order $(142536)$. Then we get the matrix with three $2\times 2$
blocks on the main diagonal.
Therefore, the determinant of $M(i\om)$ is  product of the determinants
of the three matrices. Further,
\be\la{inM}
M^{-1}(i\om)=\left(\ba{ll}
M_{11}(\om)  &  M_{12}(\om) \\
M_{21}(\om)  &  M_{22}(\om)
\ea\right)
\ee
where
$$
M_{11}(\om)=M_{22}(\om)
=\left(\ba{ccc}
\frac{-i\om\gamma^3}{\om^2\gamma^3+F_{11}(\om)}            &  0    & 0   \\
0  & \frac{-i\om\gamma}{\om^2\gamma+F_{22}(\om)}               & 0   \\
0  &           0   &  \frac{-i\om\gamma}{\om^2\gamma+F_{33}(\om)}
\ea\right)
$$
$$
M_{12}
=\left(\ba{ccc}
\frac{-1}{\om^2\gamma^3+F_{11}}             &  0    & 0   \\
0  & \frac{-1}{\om^2\gamma+F_{22}}                & 0   \\
0  &           0   &  \frac{-1}{\om^2\gamma+F_{33}} 
\ea\right),~~
M_{21}
=\left(\ba{ccc}
\frac{-\gamma F_{11}}{\om^2\gamma^3+F_{11}}            &  0    & 0   \\
0  & \frac{-\gamma F_{22}}{\om^2\gamma+F_{22}}               & 0   \\
0  &           0   &  \frac{-\gamma F_{33}}{\om^2\gamma+F_{33}}
\ea\right)
$$
Let us prove that for $\om\in(-\mu,\mu)$
\be\la{fjj}
F_{jj}(\om)=\om^2f_{jj}(\om),\quad
f_{jj}(\om)\in\C^{\infty}(-\mu,\mu), \quad f_{jj}(0)>0
\ee
Indeed,
formula (\ref{Fjj}) implies that $F_{jj}(0)=0$. Differentiating (\ref{Fjj}), we obtain
$$
F_{jj}'(0)=2\int k_j^2{\cal B}(k)dk\frac{|v|k_1}
{(k^2+m^2-(|v|k_1)^2)^2}=0
$$
since integrand is odd function in respect to $k_1$, and
$$
F_{jj}''(0)=2\int k_j^2{\cal B}(k)dk\frac{k^2+m^2+3(|v|k_1)^2}
{(k^2+m^2-(|v|k_1)^2)^3}>0
$$
By (\ref{fjj}) we can represent the matrices $M_{ij}(\om)$ as
$$
M_{11}(\om)=M_{22}(\om)
=\frac {1}{\om}\left(\ba{ccc}
\frac{-i\gamma^3}{\gamma^3+f_{11}(\om)}            &  0    & 0   \\
0  & \frac{-i\gamma}{\gamma+f_{22}(\om)}               & 0   \\
0  &           0   &  \frac{-i\gamma}{\gamma+f_{33}(\om)}
\ea\right)=\frac {1}{\om}{\cal M}_{11}(\om)
$$
\be\la{cMjj}
M_{12}(\om)=\frac {1}{\om^2}
\left(\ba{ccc}
\frac{-1}{\gamma^3+f_{11}(\om)}            &  0    & 0   \\
0  & \frac{-1}{\gamma+f_{22}(\om)}               & 0   \\
0  &           0   &  \frac{-1}{\gamma+f_{33}(\om)}
\ea\right)=\frac {1}{\om^2}{\cal M}_{12}(\om)
\ee
$$
M_{21}(\om)
=\left(\ba{ccc}
\frac{-\gamma^3f_{11}(\om)}{\gamma^3+f_{11}(\om)}            &  0    & 0   \\
0  & \frac{-\gamma f_{22}(\om)}{\gamma+f_{22}(\om)}               & 0   \\
0  &           0   &  \frac{-\gamma f_{33}(\om)}{\gamma+f_{33}(\om)}
\ea\right)={\cal M}_{21}(\om)
$$
where ${\cal M}_{ij}(\om)\in C^{\infty}(-\mu,\mu)$.
\section{Symplectic orthogonality conditions}
Here we derive  conditions (\ref{soin})  from the symplectic orthogonality conditions
(\ref{orth}). First let us compute $\Phi(0)$. Formulas  (\ref{PsiP}) and (\ref{Philam}) imply
$$
(\Phi(0))_j=\langle \hat G_0^{11}\hat\Psi_{01}
+\hat G_0^{12}\hat\Psi_{02},ik_j\hat\rho_1\rangle
+\langle\hat G_0^{11}\hat\Psi_{02}- \hat G_0^{12}\hat\Psi_{01},ik_j\hat\rho_2\rangle,
\quad j=1,2,3
$$
On the other hand, by  (\ref{Delam}) formulas  (\ref{system}) read
$$
\hat\psi_{v1}=-\hat G_0^{11}\hat\rho_2+\hat G_0^{12}\hat\rho_1,\quad
\hat\psi_{v2}=\hat G_0^{11}\hat\rho_1+\hat G_0^{12}\hat\rho_2
$$
Hence, for $j=1,2,3$
\beqn\nonumber
0&=&-\Om(Z_0,\tau_j)=\langle\Psi_{01},\pa_j\psi_{v2}\rangle-
\langle\Psi_{02},\pa_j\psi_{v1}\rangle+P_0\cdot e_j\\
\nonumber
&=&-\langle\Psi_{01},ik_j(\hat G_0^{11}\hat\rho_1+\hat G_0^{12}\hat\rho_2)\rangle
+\langle\Psi_{02},ik_j(\hat G_0^{12}\hat\rho_1-\hat G_0^{11}\hat\rho_2)\rangle
+P_0\cdot e_j=(\Phi(0)+P_0)_j
\eeqn
since $(\hat G_0^{11})^*=-\hat G_0^{11}$, $(\hat G_0^{12})^*=\hat G_0^{12}$.
Hence the first condition (\ref{soin}) follows.
Further, 
$$
\pa_{\lam}\hat G_{\lam}^{11}\Big|_{\lam=0}=\frac{-1-2ivk\hat G_{0}^{11}}{k^2+m^2-(vk)^2},
\quad  \pa_{\lam}\hat G_{\lam}^{12}\Big|_{\lam=0}=\frac{-2ivk\hat G_{0}^{12}}{k^2+m^2-(vk)^2}
$$
Then (\ref{PsiP}) and (\ref{Philam}) imply for $j=1,2,3$
$$
(\Phi'(0))_j\!=\!-\Big\langle \frac{\hat\Psi_{01}\!+\!2ivk(\hat G_{0}^{11}\hat\Psi_{01}+
\hat G_{0}^{12}\hat\Psi_{02})}{k^2+m^2-(vk)^2},ik_j\hat\rho_1\Big\rangle
\!-\!\Big\langle\frac{\hat\Psi_{02}\!+\!2ivk(\hat G_{0}^{11}\hat\Psi_{02}\!-\!
 G_0^{12}\hat\Psi_{01})}{k^2+m^2-(vk)^2},ik_j\hat\rho_2\Big\rangle
$$
On the other hand, from (\ref{system}) and (\ref{derv}) it follows that  for $j=1,2,3$
$$
\pa_{v_j}\hat\psi_{v1}=\frac{i k_j\hat\rho_2+2k_jvk
(-\hat G_0^{11}\hat\rho_2+\hat G_0^{12}\hat\rho_1)}{k^2+m^2-(vk)^2},\quad
\pa_{v_j}\hat\psi_{v2}=\frac{-i k_j\hat\rho_1+2k_jvk
(\hat G_0^{11}\hat\rho_1+\hat G_0^{12}\hat\rho_2)}{k^2+m^2-(vk)^2}
$$
Hence,
\beqn\nonumber
0\!\!\!&=&\!\!\!\Om(Z_0,\tau_{j+3})=\langle\Psi_{01},\pa_{v_j}\psi_{v2}\rangle
-\langle\Psi_{02},\pa_{v_j}\psi_{v1}\rangle+Q_0\cdot\pa_{v_j}p_v\\
\nonumber
\!\!\!&=&\!\!\!\Big\langle \Psi_{01},\frac{-ik_j\hat\rho_1+2k_jvk
(\hat G_0^{11}\hat\rho_1+\hat G_0^{12}\hat\rho_2)}{k^2+m^2-(vk)^2}\Big\rangle
-\Big\langle \Psi_{02},\frac{ik_j \hat\rho_2+2k_jvk
(-\hat G_0^{11}\hat\rho_2+\hat G_0^{12}\hat\rho_1)}{k^2+m^2-(vk)^2}\Big\rangle
\\
\nonumber
\!\!\!&+&\!\!\!Q_0\cdot\pa_{v_j}p_v=(\Phi'(0)+B_v^{-1}Q_0)_j,\quad j=1,2,3
\eeqn
since $Q_0\cdot\pa_{v_j}p_v=Q_0\cdot B_v^{-1}e_j=B_v^{-1}Q_0\cdot e_j$.
Hence the second condition (\ref{soin}) follows.


\begin{thebibliography}{18}

\bibitem{A} 
M.Abraham, Theorie der Elektrizitat, Band 2:
Elektromagnetische Theorie der Strahlung, Teubner, Leipzig (1905).

\bibitem{BP1}
V.S.Buslaev, G.S.Perelman, Scattering for the nonlinear Schr\"odinger equation:
states close to a soliton, {\em St.Petersburg Math. J.} {\bf 4} (1993), 1111-1142.

\bibitem{BP2}
V.S.Buslaev, G.S.Perelman, On the stability of solitary waves for nonlinear
Schr\"{o}dinger equations, {\em Trans. Amer. Math. Soc.} {\bf 164} (1995), 75-98.

\bibitem{BS}
V.S.Buslaev, C.Sulem, On asymptotic stability of solitary waves for nonlinear
Schr\"odinger equations, {\em Ann. Inst. Henri Poincar\'e, Anal.
Non Lin\'eaire}  {\bf  20}, no.3 (2003), 419-475.

\bibitem{Cu}
S.Cuccagna, Stabilization of solutions to nonlinear Schr\"odinger equations,
{\em Commun. Pure Appl. Math.}  {\bf  54},  no.9  (2001), 1110-1145.

\bibitem{EGS}
M. Esteban, V. Georgiev, E. Sere, Stationary solutions of the Maxwell-Dirac
and the Klein-Gordon-Dirac equations, {\em Calc. Var. Partial Differ. Equ.}
{\bf  4}, no.3 (1996), 265-281.

\bibitem{IKM}
V.Imaikin, A.Komech, N.Mauser, Soliton-type asymptotics for the coupled Maxwell-Lorentz
equations, {\em Ann. Inst. Poincar\'e, Phys. Theor.} {\bf 5} (2004), 1117-1135.

\bibitem{ikv}
V.Imaikin, A.Komech, B.Vainberg, On scattering of solitons for the Klein-Gordon equation
coupled to a particle, {\em Comm. Math. Phys.} {\bf 268}, no.2 (2006), 321-367.

\bibitem{EKS}
A.I.Komech, Linear Partial Differential Equations with Constant Coefficients, p.127-260
in: Yu.V.Egorov, A.I.Komech, M.A.Shubin, Elements of the Modern Theory
of Partial Differential Equations, Springer, Berlin, 1999.

\bibitem{KK}
A.Komech, E.Kopylova, Scattering of solitons for Schr\"odinger equation coupled to a
particle, {\em Russian J. Math. Phys.} {\bf 13}, no.2 (2006), 158-187. 

\bibitem{KS}
A.I.Komech, H.Spohn, Soliton-like asymptotics for a classical particle interacting with
a scalar wave field, {\em Nonlin. Analysis} {\bf 33} (1998), 13-24.

\bibitem{KSK}  A. Komech, H. Spohn,  M. Kunze, Long-time asymptotics for a
classical particle interacting with a scalar wave field, {\em Comm. Partial
Differential Equations}{\bf \thinspace \thinspace 22} (1997), 307-335.

\bibitem{Sig}
I.M.Sigal, Nonlinear wave and Schr\"odinger equations. I. Instability of periodic and
quasiperiodic solutions, {\em   Comm. Math. Phys.} {\bf 153},  no. 2 (1993), 297--320.

\bibitem{SW3}
A.Soffer, M.I.Weinstein, Resonances, radiation damping and instability in Hamiltonian
nonlinear wave equations,  {\em Invent. Math.} {\bf  136}, no. 1 (1999), 9-74.

\bibitem{SW4}
A.Soffer, M.I.Weinstein, Selection of the ground state for nonlinear Schr\"odinger equations,
{\em Rev. Math. Phys.} {\bf 16}, no. 8 (2004), 977-1071.

\end{thebibliography}
\end{document}